\documentclass[11pt,letterpaper]{article}

\usepackage[top=1in, bottom=1in, left=1in, right=1in]{geometry}

\usepackage{amsmath}
\usepackage{amsthm}
\usepackage{amsfonts}
\usepackage{amssymb}
\usepackage{bbm}
\usepackage{mathabx} 
\usepackage{mathrsfs}
\usepackage{graphicx}
\usepackage{color}
\usepackage[inline]{enumitem}
\newcommand{\BIF}{\textup{B},\infty,\texttt{F}}
\newcommand{\BIO}{\textup{B},\infty,1}
\newcommand{\E}[1]{\textup{E}\big( #1 \big)}
\newcommand{\hE}[1]{\hat{\textrm{E}}\big( #1 \big)}
\newcommand{\bH}{\mb{H}}
\newcommand{\bHhat}{\hat{\mb{H}}}
\newcommand{\res}[2]{\big[ #1 \big]_{#2}}

\newcommand{\T}{\intercal}
\newcommand{\supp}{\textup{S}_i}
\newcommand{\ind}{\textup{ind}}
\newcommand{\budget}{\textup{b}}
\newcommand{\prob}{\textup{Pr}}

\newcommand{\real}{\mathbb{R}}
\newcommand{\graph}{\mathcal{G}}
\newcommand{\nodes}{\mathcal{V}}
\newcommand{\edges}{\mathcal{E}}
\newcommand{\actions}{\mathcal{A}}
\newcommand{\game}{\mathscr{G}}
\newcommand{\nash}{\textup{NE}}

\newcommand{\f}{\textup{vec}}

\newcommand{\mb}[1]{\mathbf{#1}}
\newcommand{\minus}{\text{\normalfont -}}
\newcommand{\seq}[1]{\{ 1,\dots, #1 \}}
\newcommand{\inner}[2]{\langle #1 \,, #2 \rangle}

\newtheorem{theorem}{Theorem} 
 
\newtheorem{lemma}{Lemma}
\newtheorem{assumption}{Assumption}

\title{Provable Computational and Statistical Guarantees for Efficient Learning of Continuous-Action Graphical Games}
\date{}

\author{%
	Adarsh Barik \\
	Department of Computer Science\\
	Purdue University\\
	West Lafayette, Indiana, USA\\
	\texttt{abarik@purdue.edu} \\
	\and
	Jean Honorio \\
	Department of Computer Science \\
	Purdue University \\
	West Lafayette, Indiana, USA\\
	\texttt{jhonorio@purdue.edu} \\
}

\begin{document}

\maketitle

\begin{abstract}
  In this paper, we study the problem of learning the set of pure strategy Nash equilibria and the exact structure of a continuous-action graphical game with quadratic payoffs by observing a small set of perturbed equilibria. A continuous-action graphical game can possibly have an uncountable set of Nash euqilibria. We propose a $\ell_{12}-$ block regularized method which recovers a graphical game, whose Nash equilibria are the $\epsilon$-Nash equilibria of the game from which the data was generated (true game). Under a slightly stringent condition on the parameters of the true game, our method recovers the exact structure of the graphical game. Our method has a logarithmic sample complexity with respect to the number of players. It also runs in polynomial time.
\end{abstract}

\section{INTRODUCTION}
\label{sec:introduction}

The real world is filled with scenarios which arise due to competitive actions by selfish individual players who are trying to maximize their own utilities or payoffs. Non-cooperative game theory has been considered as the appropriate mathematical framework to formally study \emph{strategic} behavior in such multi-agent scenarios. In such scenarios, each agent decides its action based on the actions of other players. The core solution concept of \emph{Nash equilibrium (NE)}~\cite{Nash51} serves a descriptive role of the stable outcome of the overall behavior of self-interested agents (e.g., people, companies, governments, groups or autonomous systems) interacting strategically with each other in distributed settings.

\paragraph{Graphical Games.}

The introduction of compact representations to game theory over the last two decades have extended algorithmic game theory's potential for large-scale, practical applications often encountered in the real world. Introduced within the AI community about two decades ago, \emph{graphical games}~\cite{Kearns01} constitute an example of one of the first and arguably one of the most influential graphical models for game theory. Indeed, graphical games played a prominent role in establishing the computational complexity of computing NE in normal-form games as well as in succinctly representable multiplayer games (see, e.g.,~\cite{Daskalakis06,Daskalakis09,Daskalakis09b} and the references therein). 

Players can take actions in either a discrete space (for example in voting) or in a continuous space (for example in simultaneous auctions in online advertising). Correspondingly, graphical games can be studied in both domains. In this paper, we focus on continuous-action graphical games.

\paragraph{Inference in Graphical Games.}

There has been considerable progress on \emph{computing} classical equilibrium solution concepts such as NE and \emph{correlated equilibria}~\cite{Aumann74} in graphical games (see, e.g.,~\cite{Blum06,Kearns01,Kakade03,Ortiz02,Papadimitriou08,Vickrey02} and the references therein) as well as on computing the \emph{price of anarchy} in graphical games (see, e.g.,~\cite{Benzwi11}). \cite{Irfan14} identified the most influential players, i.e., a small set of players whose collective behavior forces every other player to a unique choice of action.  All the work above focus on inference problems for graphical games, and fall in the field of algorithmic game theory.

\paragraph{Learning Graphical Games.}

The aforementioned problems of computing Nash equilibrium, correlated equilibrium or price of anarchy often assume that the structure and payoffs of the games under consideration are already available. Relatively less attention has been paid to the problem of \emph{learning} both the structure and payoffs of graphical games from data. Addressing this problem is essential to the development, potential use and success of game-theoretic models in practical applications. In this paper, we study the problem of learning the complete characterization of pure strategy Nash equilibrium and structure of the graph in a continuous-action graphical game.

\paragraph{Related Work.}

There has been considerable amount of work done for learning games in the discrete-action setting. \cite{Honorio15} proposed a maximum-likelihood approach to learn “linear influence games” - a
class of parametric graphical games with binary actions and linear payoffs.  However, their method runs in exponential time and the authors assumed a specific observation model for the strategy profiles. For the same specific observation model, \cite{Ghoshal16} proposed a polynomial time algorithm, based on $\ell_1$-regularized logistic regression, for learning linear influence games. Their strategy profiles (or joint actions) were drawn from a mixture of uniform distributions: one over the pure-strategy Nash equilibria (PSNE) set, and the other over its complement. \cite{Ghoshal17} obtained necessary and sufficient conditions for learning linear influence games under arbitrary observation model. \cite{Garg16} use a discriminative, max-margin based approach, to learn tree structured polymatrix games\footnote{Polymatrix games are graphical games where each player's utility is a sum of unary (single player) and pairwise (two players) potential functions.}. Their method runs in exponential time and the authors show that learning polymatrix games is NP-hard under this max-margin setting, even when the class of graphs is restricted to trees. Finally, \cite{Ghoshal18} proposed a polynomial time algorithm for learning sparse polymatrix games in the discrete-action setting.

Regarding inference for continuous-action games, \cite{facchinei2007generalized} and \cite{scutari2010convex} provide a survey of variational inequality methods and Gauss-Seidel methods to compute generalized Nash equilibrium for pure strategy games. \cite{perkins2013asynchronous} and \cite{perkins2015mixed} studied a mixed-strategy actor-critic algorithm which converges to a probability distribution that assigns most weight to equilibrium states. \cite{mertikopoulos2019learning}  provided sufficient conditions under which no-regret learning converges to equilibrium. 

Continuous-action games with quadratic payoffs have been used extensively in the game theory literature \cite{ballester2006s,bramoulle2014strategic,galeotti2017targeting,acemoglu2015networks}.   \cite{leng2018learning} proposed algorithms to learn games with quadratic payoffs, in a simplified setting. However, the authors do not provide any theoretical guarantees. In this work, we focus on provable guarantees for a far more general class of games with quadratic payoffs in the high-dimensional regime. 

\paragraph{Our Contribution.}

We aim to propose a novel method to learn graphical games with quadratic payoffs, with the following provable guarantees in mind:
\begin{enumerate*}
	\item {\bf Correctness -} We want to develop a method which correctly recovers the set of Nash equilibria and the structure of the graphical games.
	\item {\bf Computational efficiency -} Our method must run fast enough to handle the high dimensional cases. Ideally, we want to have polynomial time complexity with respect to the number of players.
	\item {\bf Sample complexity -} We would like to use as few samples as possible for recovering the set of Nash equilibria. We want to achieve logarithmic sample complexity with respect to the number of players.
\end{enumerate*}

To this end, we propose a block-norm regularized method to learn graphical games with quadratic payoff functions. For $n$ players, at most $d$ in-neighbors per player, and $k$-dimensional action vectors, we show that $\mathcal{O}( k^5 d^3 \log (d n k) )$  \emph{samples} are sufficient to recover the complete characterization of the set of $\epsilon$-Nash equilibria. Under slightly more stringent conditions, we also recover the true structure of the game. Our method also runs in polynomial time complexity. 

Regarding the main challenges that we address, first, the set of Nash equilibria for continuous-action games is uncountable, while for discrete-action games is countable.  Our method provides the complete characterization of such uncountable sets. Our method is also oblivious to the exact process under which players converge to Nash equilibria. In fact, Nash equilibria can be "chosen" by nature in an arbitrary non-probabilistic fashion. We also do not assume any particular process that ``chooses'' Nash equilibria, such as, for instance, a stochastic process. Our method only needs access to some small number perturbed  equilibria.

\section{PRELIMINARIES}
\label{sec:notation and problem setting}

In this section, we introduce our notation and formally define the problem of learning graphical games with quadratic utility functions. Consider a directed graph $\graph(\nodes,\edges)$, where $\nodes$ and $\edges$ are set of vertices and edges respectively. We define $\nodes \triangleq \seq{n}$, where each vertex corresponds to one player. We denote the in-neighbors of a player $i$ by $\supp$, i.e., $\supp = \{ j \mid (j, i) \in \edges \}$ (i.e., the set of nodes that point to node $i$ in the graph). All the other players are denoted by $\supp^c$, i.e., $\supp^c = \seq{n} \backslash (\supp \cup i) $. Let $|\supp| \leq d$ and $|\supp^c| \leq n$. 

For each player $i \in \nodes$, there is a set of actions or \emph{pure-strategies}  $\actions_i$. That is, player $i$ can take action $x_i \in \actions_i$. Each action $x_i$ consists of making $k$ decisions on a limited budget $\budget \in \real$. We consider games with continuous actions. Mathematically, $x_i \in \real^k$ and $\| x_i \|_2 \leq \budget$. For each player $i$, there is also a local payoff function $u_i : \actions_i \times ( \bigtimes_{j \in \supp} \actions_j ) \to \real$ mapping the joint action of player $i$ and its in-neighbors $\supp$, to a real number. Later, we will define a particular kind of local payoff function which is of our interest. A joint action $\mb{x}^* \in \bigtimes_{i \in \nodes} \actions_i$ is a \emph{pure-strategy Nash equilibrium (PSNE)} of a graphical game iff, no player $i$ has any incentive to unilaterally deviate from the prescribed action $x_i^* \in \actions_i$, given the joint action of 
its in-neighbors $x_{\supp}^* \in \bigtimes_{j \in \supp} \actions_j$ in the equilibrium. We denote a game by $\game$, and the set of all PSNE and $\epsilon$-PSNE of $\game$, by  $\nash(\game)$ and $\nash_{\epsilon}(\game)$ respectively, for a constant $\epsilon > 0$. Formally,
\begin{align*}
	\nash(\game) &\triangleq \{ \mb{x}^* \in \bigtimes_{i \in \nodes} \actions_i \mid x_i^* \in \arg \max_{x_i \in \actions_i} u_i(x_i, x_{\supp}^*), \forall i \in \nodes  \} \\
	\nash_{\epsilon}(\game) &\triangleq \{ \mb{x}^* \in \bigtimes_{i \in \nodes} \actions_i \mid u_i(x_i^*, x_{\supp}^*) \geq -\epsilon + \max_{x_i \in \actions_i} u_i(x_i, x_{\supp}^*) , \forall i \in \nodes  \}
\end{align*}

\paragraph{Parametric Payoffs.} 
We are interested in solving a parametrized version of the problem. In that, given the weights $W_{ij}^* \in \real^{k\times k}, \forall i, j \in \nodes$, for each player $i$, we define the set of in-neighbors of player $i$ as $\supp = \{ j \mid W_{ij}^* \ne 0 \}$ and the payoff function 
\begin{align*} 
u_i(x_i, x_{\supp}) = - \| x_i - \sum_{j \in \supp} W_{ij}^* x_j \|_2
\end{align*}
Consider $\max_{x_i} u_i(x_i, x_{\supp}^*) = 0, \forall i \in \seq{n} $, then in a PSNE, each player $i$ matches their action $x_i$ to the weighted actions of their neighbors, i.e., $\sum_{j \in \supp} W_{ij}^* x_j^*$. Let $\epsilon > 0$ be a constant. The set of all $\epsilon$-PSNE of $\game$ is
\begin{align*} 
\nash_{\epsilon}(\game) =& \{ \mb{x}^* \in \bigtimes_{i=1}^n \actions_i \mid \| x_i^* - \sum_{j \in \supp} W_{ij}^* x_j^*  \|_2 \leq \epsilon, \forall i \in \nodes \}.
\end{align*} 

\paragraph{Sampling.} Given the above characterization, the set of $\epsilon$-PSNE is a convex polytope. We observe samples from the set of noisy PSNE which follow a local noise mechanism that adds noise independently per player. Observed joint actions 
\begin{align*}
\mb{x} = \mb{x}^* + \mb{e}
\end{align*}  
where $\mb{x}^*$ is a Nash equilibrium, that is $\mb{x}* \in \nash(\game)$ and $\mb{e}$ is independent zero mean sub-Gaussian noise with variance proxy $\sigma^2$. The class of sub-Gaussian variates includes for instance Gaussian variables, any bounded random variable (e.g. Bernoulli, multinomial, uniform), any random variable with strictly log-concave density, and any finite mixture of sub-Gaussian variables. 

\paragraph{Norms and Notations.}
For a matrix $\mb{A} \in \real^{p \times q}$ and two sets $S \subseteq \seq{p}$ and $T \subseteq \seq{q}$,  $\mb{A}_{S T}$ denotes $\mb{A}$ restricted to rows in $S$ and columns in $T$.  Similarly, $\mb{A}_{S.}$ and $\mb{A}_{.T}$ are row and column restricted matrices respectively. For a vector $\textbf{m} \in \real^q$, the $\ell_{\infty}$-norm is defined as $ \| \textbf{m} \|_{\infty} = \max_{i \in \seq{p}} | \textbf{m}_i | $. The Frobenius norm for a matrix $\mb{A} \in \real^{p \times q}$  is defined as 
\begin{align*}
 \| \mb{A} \|_F =  \sqrt{\sum_{i=1}^p \sum_{j=1}^q| \mb{A}_{ij} |^2}.
 \end{align*} 
 The $\ell_{\infty}$-operator norm for $\mb{A}$ is defined as 
\begin{align*} 
  \| \mb{A} \|_{\infty, \infty} = \max_{i \in \seq{p}}\sum_{j=1}^q | \mb{A}_{ij} | .
\end{align*} 
The spectral norm of $\mb{A}$ is defined as
\begin{align*}
 \| \mb{A} \|_{2,2} = \sup_{ \|\mb{x}\|_2 = 1} \|\mb{A} \mb{x}\|_2.
\end{align*} 
We also define a block matrix norm for row-partitioned block matrices. Let $\mb{A} \in \real^{\sum_{i=1}^k p_i \times q}, \forall i \in \seq{k}$ be a row-partitioned block matrix defined as follows: $ \mb{A} = \begin{bmatrix} \mb{A}_1 & \cdots & \mb{A}_k \end{bmatrix}^\T$ where each $\mb{A}_i \in \real^{p_i \times q}$.
Then 
\begin{align*} 
\| \mb{A} \|_{\BIF} &= \max_{i \in \seq{k}} \| \mb{A}_i \|_F\\
 \| \mb{A} \|_{\BIO} &= \max_{i \in \seq{k}} \sum_{l=1,m=1}^{l=p_i,m=q} | [\mb{A}_i]_{lm} |.
\end{align*}

\section{MAIN RESULT}
\label{sec:main result}

In this section, we describe our main theoretical results. But before we do that, we discuss some technical assumptions which are needed for our proofs.

\begin{assumption}[Budgeted actions]
	\label{assum:budgeted_actions}
	For all $x_i \in \actions_i, \| x_i \|_2 \leq \budget, \forall i \in \seq{n}$ for some $\budget > 0$.
\end{assumption}
\begin{assumption}[Maximum zero utility]
	\label{assum:max_zero_utility}
	At PSNE, $u_i(x_i^*, x_{\minus i}^*) = 0, \forall i \in \seq{n}$.
\end{assumption}
\begin{assumption}[Mutual Incoherence]
	\label{assum:mutual_incoherence}
	Consider $\bH = \frac{1}{T} \sum_{t=1}^T \big( {\mb{x}_{\minus i}^*}^t {{\mb{x}_{\minus i}^*}^t}^\T +  \sigma^2 \mb{I}\big)$ where $\mb{I}$ is the identity matrix, then $\res{\bH}{\supp^c\supp} \res{\bH}{\supp\supp}^{\minus 1} \leq 1 - \alpha $ for some $\alpha \in (0, 1]$.
\end{assumption}

Assumption \ref{assum:budgeted_actions} simply states that each player has a limited budget to allocate for its actions. For instance, consider simultaneous auctions in an online advertising, where a company chooses how to allocate its budget into several options. For a sufficiently large budget $\budget$, Assumption \ref{assum:max_zero_utility} is not difficult to fulfill for quadratic payoffs. We propose a mutual incoherence assumption (Assumption \ref{assum:mutual_incoherence}) for games. While mutual incoherence is new to graphical games, it has been a standard assumption in various estimation problems such as compressed sensing~\cite{wainwright2009sharp}, Markov random fields~\cite{ravikumar2010high}, non-parametric regression~\cite{ravikumar2007spam}, diffusion networks~\cite{daneshmand2014estimating}, among others.     

Now that all our assumptions are in place, we are ready to setup our estimation problem. Consider that we have access to $T$ perturbed equilibria, i.e., we have access to $x_i^t = {x_i^*}^t + e_i^t$ where superscript $t$ denotes the $t$-th sample and $e_i^t \in \real^k$ is a vector of zero-mean mutually independent sub-Gaussian noises with variance proxy $\sigma^2$. We estimate the parameters $W_{ij}$ for each $i, j \in \seq{n}$ by solving the following optimization problem:   
\begin{align}
\label{eq:opt_prob}
	\hat{W_{i\cdot}} =& \arg\min_{W_{i\cdot}} \frac{1}{T} \sum_{t=1}^{T} \| x_i^{t} - \sum_{\substack{j=1\\ j \ne i}}^n W_{ij} x_j^{t} \|_2^2 + \lambda \sum_{\substack{j=1\\ j \ne i}}^n \| W_{ij} \|_F 
\end{align}  
where $W_{i\cdot}$ denotes the collection of all $W_{ij}, \forall j \in \seq{n}, j \ne i$. Our next theorem states that the recovered $\hat{W}_{i.}$ completely characterizes the set of all $\epsilon$-Nash equilibria.

\begin{theorem}
	\label{thm:eps_nash_equilibria}
	Consider a continuous-action graphical game $\game$ such that Assumptions \ref{assum:budgeted_actions}, \ref{assum:max_zero_utility} and \ref{assum:mutual_incoherence} are satisfied for each player. Let $ \lambda > \max( 24\sqrt{2} \frac{1-\alpha}{\alpha} \sigma \budget W_{\max} \sqrt{\frac{k |\supp| \log( 2k^2 |\supp|)}{T}},  192 \frac{1-\alpha}{\alpha} \sigma^2  W_{\max} \\ \sqrt{\frac{k \log(k^2 |\supp|)}{T}}, 192 \frac{1-\alpha}{\alpha} \sigma^2 \sqrt{k} W_{\max}  \sqrt{\frac{W_{\max} |\supp| \log(|\supp|k)}{TW_{\min}}},  192 \frac{1-\alpha}{\alpha} k^{\frac{1}{4}} \sigma \sqrt{\frac{\log (2k^2|\supp|)}{T}},  \frac{24\sqrt{2}}{\alpha} k \sigma \budget W_{\max}  \\ \sqrt{\frac{|\supp^c \log(2k^2|\supp^c|)|}{T}},  \frac{192}{\alpha} \sigma^2 k W_{\max}  \sqrt{\frac{\log(k^2 |\supp^c|)}{T}}, \frac{192}{\alpha} \sigma^2 k W_{\max}  \sqrt{\frac{W_{\max} |\supp^c|^{\frac{1}{2}} \log(|\supp^c|k) }{T}},   \frac{24\sqrt{2}}{\alpha} k \sigma \budget \sqrt{\frac{\log (2k^2 |\supp^c|)}{T}} , \\ \frac{192}{\alpha} \sigma \sqrt{\frac{k \log(2k^2|\supp^c|)}{T}}, \frac{24 (1-\alpha) \sigma^2 \sqrt{k} \max_{ij} |W_{ij}| }{\alpha}, \frac{24 \sigma^2 k \max_{ij} |W_{ij}| }{\alpha} ) $, then the following claims hold.
	\begin{enumerate}
		\item We can recover $\nash_\epsilon(\game)$ by estimating $\mb{W}$ from the optimization problem \eqref{eq:opt_prob}.
		\item Furthermore, for each player $i \in \seq{n}$, if $\min_{j \in \supp} \|W_{ij}^*\|_{F} > 2 \delta(k, |\supp|, C_{\min}, \alpha, \lambda, \sigma, W_{\max})$, then we recover the exact structure of the graphical game $\game$. 
	\end{enumerate} 	
	where $C_{\min}$ is the minimum eigenvalue of $\res{\bH}{\supp\supp}$, $W_{\max} \triangleq \max_{i,j} |\mb{W}_{ij}|$, $\epsilon = |\supp| \delta(k, |\supp|, C_{\min}, \alpha, \\ \lambda, \sigma, W_{\max}) \budget$ and 
	
	\begin{align*}
	&\delta(k, |\supp|, C_{\min}, \alpha, \lambda, \sigma, W_{\max}) = k \sqrt{k|\supp|} \frac{2}{C_{\min}} 
	(\frac{\alpha \lambda}{24(1-\alpha)} + \sigma^2 (1 + \frac{\alpha \lambda}{24 (1-\alpha)\sigma^2 \sqrt{k} W_{\max}} +\\
	& \frac{\alpha \lambda}{24 (1 - \alpha) \sigma^2 \sqrt{k} \max_{ij} W_{\max}}) \sqrt{k} W_{\max}) + k \sqrt{k|\supp|}\frac{2}{C_{\min}} (\frac{\alpha \lambda}{24 (1-\alpha)} + \frac{\alpha \lambda}{24 (1-\alpha)})  +   \frac{\lambda}{2}  k \sqrt{k|\supp|}\frac{2}{C_{\min}} . 
	\end{align*}
\end{theorem}
\begin{proof}
	We will make use of the \emph{primal-dual witness} method to prove Theorem \ref{thm:eps_nash_equilibria}. By using the definition of Frobenious norm, optimization problem \eqref{eq:opt_prob} can be equivalently written as
\begin{align}
	\label{eq:form0}
	\begin{split}
		\hat{W_{i\cdot}} =& \arg\min_{W_{i.}} \frac{1}{T} \sum_{t=1}^{T} \| x_i^{t} - \sum_{j=1, j \ne i}^n W_{ij} x_j^{t} \|_2^2 + \lambda \sum_{j=1, j \ne i}^n  \sup_{\| Z_{ij} \|_F \leq 1}\inner{Z_{ij}}{W_{ij}}  \\
		\end{split}
\end{align}
Consider the term $\sup_{\| Z_{ij} \|_F \leq 1}\inner{Z_{ij}}{W_{ij}}$. We can assign specific values to $Z_{ij}$ to get the maximum possible value of $\inner{Z_{ij}}{W_{ij}}$. In particular, we can take if $ W_{ij} \ne \mb{0}$ then $Z_{ij} = \frac{W_{ij}}{\| W_{ij} \|_F}$, and if $W_{ij} = \mb{0}$ then $\|Z_{ij}\|_F \leq 1$.
Note that in the first case $\| Z_{ij} \|_F = 1$ and thus it gives the maximum value for $\inner{Z_{ij}}{W_{ij}}$ and no further improvement is possible. In the second case, since $W_{ij} = \mb{0}$, $Z_{ij}$ can take any value such that  $\|Z_{ij}\|_F \leq 1$ without affecting $\inner{Z_{ij}}{W_{ij}}$. We fix $Z_{ij}$ to one such value and rewrite equation \eqref{eq:form0} as
\begin{align}
\label{eq:form1}
\begin{split}
		\hat{W_{i\cdot}} &= \arg\min_{W_{i.}} \frac{1}{T} \sum_{t=1}^{T} \| x_i^{t} - \sum_{j=1, j \ne i}^n W_{ij} x_j^{t} \|_2^2 + \lambda \sum_{j=1, j \ne i}^n  \inner{Z_{ij}}{W_{ij}}  
	\end{split}
\end{align}
where the last equality comes by keeping in mind that $Z_{ij}$ are chosen as described above.  We can rewrite equation \eqref{eq:form1} as,
\begin{align*}
\begin{split}
	\hat{W_{i\cdot}} =& \arg\min_{W_{i\cdot}} \frac{1}{T} \sum_{t=1}^T \big( -2 \sum_{j=1, j \ne i}^n {x_i^t}^{\T} W_{ij} x_j^t + \sum_{\substack{j=1, j \ne i \\ k=1, k \ne i}}^n {x_j^t}^{\T}  W_{ij}^\T W_{ik} x_k^t \big) + \lambda \sum_{j=1, j\ne i}^n \inner{Z_{ij}}{W_{ij}}
\end{split}
\end{align*} 
Using the stationarity Karush-Kuhn-Tucker condition at the optimum, for each $W_{ij}$ we can  write,
\begin{align}
\begin{split}
	\frac{1}{T} \sum_{t=1}^T (-2 x_i^t {x_j^t}^\T + 2 \sum_{k=1, k \ne i}^n W_{ik} x_k^t x_j^t) + \lambda Z_{ij} = 0
\end{split}
\end{align}
Further note that $x_i^t = {x_i^*}^t + e_i^t$ , ${x_i^*}^t = \sum_{j=1, j\ne i}^n W_{ij}^* {x_j^*}^t$ and ${x_j^*}^t = x_j^t - e_j^t $ where $e_i^t, e_j^t \in \real^k$ are zero mean sub-Gaussian vectors with variance proxy $\sigma^2$. Therefore, $x_i^t$ can be written as a function of $W_{ij}^*, x_j^t, e_i^t$ and $e_j^t$.  Thus, by substituting $x_i^t$ and writing the system of equations in vector form, we get
\begin{align}
\label{eq:stationarity}
\begin{split}
	&\frac{1}{T} \sum_{t=1}^T \big( 2 \mb{x}_{\minus i}^t {\mb{x}_{\minus i}^t}^\T (\mb{W}_{i\cdot} - \mb{W}_{i\cdot}^*) - 2 \mb{x}_{\minus i}^t {\mb{e}_{\minus i}^t}^\T \mb{W}_{i.}^* - 2 \mb{x}_{\minus i}^t {\mb{e}_i^t}^\T \big) + \lambda \mb{Z}_{i\cdot} = \mb{0}
\end{split}
\end{align}
where
\begin{align}
	\mb{W}_{i\cdot}^* = \begin{bmatrix}
	{W_{i1}^*}^\T \\ \vdots \\ {W_{in}^*}^\T
	\end{bmatrix}, \mb{W}_{i\cdot} = \begin{bmatrix}
	W_{i1}^\T \\ \vdots \\ W_{in}^\T
	\end{bmatrix}, \mb{Z}_{i\cdot} = \begin{bmatrix}
	Z_{i1}^\T \\ \vdots \\ Z_{in}^\T
	\end{bmatrix}, \mb{x}_{\minus i}^t = \begin{bmatrix}
	x_1^t \\ \vdots \\ x_n^t
	\end{bmatrix}, \quad \mb{e}_{\minus i}^t = \begin{bmatrix}
	e_1^t \\ \vdots \\ e_n^t
	\end{bmatrix}    
\end{align}
with $\mb{W}_{i\cdot}^*, \mb{W}_{i\cdot}, \mb{Z}_{i\cdot}  \in \real^{(n-1)k \times k}$ and $\mb{x}_{\minus i}^t, \mb{e}_{\minus i}^t \in \real^{(n-1)k \times 1}$.
If we denote in-neighbors of $i$ by a set $\supp$ then ${W_{ij}^*}^\T = \mb{0}$ for all $j \notin \supp$. We assume that  ${W_{ij}}^\T = \mb{0}$ for all $j \notin \supp$. This choice will be justified later. Thus, the stationarity condition can be written as,
\begin{align}
\label{eq:stationarity1}
\begin{split}
\frac{1}{T} \sum_{t=1}^T \big( 2 \res{\mb{x}_{\minus i}^t {\mb{x}_{\minus i}^t}^\T}{\cdot \supp} \res{\mb{W}_{i\cdot} - \mb{W}_{i\cdot}^*}{\supp \cdot} - 2 \res{\mb{x}_{\minus i}^t {\mb{e}_{\minus i}^t}^\T}{\supp\cdot} \res{\mb{W}_{i.}^*}{\supp\cdot} 
- 2 \mb{x}_{\minus i}^t {\mb{e}_i^t}^\T \big) + \lambda \mb{Z}_{i\cdot} = \mb{0}
\end{split}
\end{align}
Equation \eqref{eq:stationarity1} can be decomposed in two separate equations. One for the players in $\supp$ and other for players not in $\supp$ which we denote by $\supp^c$.
\begin{align}
\label{eq:support}
\begin{split}
\frac{1}{T} \sum_{t=1}^T \big( 2 \res{\mb{x}_{\minus i}^t {\mb{x}_{\minus i}^t}^\T}{\supp \supp} \res{\mb{W}_{i\cdot} - \mb{W}_{i\cdot}^*}{\supp \cdot} - 2 \res{\mb{x}_{\minus i}^t {\mb{e}_{\minus i}^t}^\T}{\supp \supp} \res{\mb{W}_{i.}^*}{\supp\cdot} 
- 2 \res{\mb{x}_{\minus i}^t}{\supp\cdot} {\mb{e}_i^t}^\T \big) + \lambda \res{\mb{Z}_{i\cdot}}{\supp\cdot} = \mb{0}
\end{split}
\end{align}
and
\begin{align}
\label{eq:nonsupport}
\begin{split}
\frac{1}{T} \sum_{t=1}^T \big( 2 \res{\mb{x}_{\minus i}^t {\mb{x}_{\minus i}^t}^\T}{\supp^c \supp} \res{\mb{W}_{i\cdot} - \mb{W}_{i\cdot}^*}{\supp \cdot} - 2 \res{\mb{x}_{\minus i}^t {\mb{e}_{\minus i}^t}^\T}{\supp^c \supp} \res{\mb{W}_{i.}^*}{\supp\cdot} 
- 2 \res{\mb{x}_{\minus i}^t}{\supp^c\cdot} {\mb{e}_i^t}^\T \big) + \lambda \res{\mb{Z}_{i\cdot}}{\supp^c\cdot} = \mb{0}
\end{split}
\end{align}

Let $\hE{.}$ denote the empirical expectation. Then equation \eqref{eq:support} can be written as,
\begin{align}
\label{eq:support0}
\begin{split}
 2 \res{\hE{\mb{x}_{\minus i} {\mb{x}_{\minus i}}^\T}}{\supp \supp} \res{\mb{W}_{i\cdot} - \mb{W}_{i\cdot}^*}{\supp \cdot} - 2 \res{\hE{\mb{x}_{\minus i} {\mb{e}_{\minus i}}^\T}}{\supp \supp} \res{\mb{W}_{i.}^*}{\supp\cdot} - 2 \hE{\res{\mb{x}_{\minus i}}{\supp\cdot} {\mb{e}_i}^\T}  + \lambda \res{\mb{Z}_{i\cdot}}{\supp\cdot} = \mb{0} 
\end{split}
\end{align}
Before we move ahead, we provide some properties of the finite-sample regime which hold with high probability. The detailed proofs of these lemmas are available in Appendix \ref{appendix:proof of theorems and lemmas}. We define $\bHhat \triangleq \hE{\mb{x}_{\minus i} {\mb{x}_{\minus i}}^\T}$, then

\begin{lemma}[Positive minimum eigenvalue]
	\label{lem:min eigenvalue}
	$\Lambda_{\min}(\res{\bHhat}{\supp \supp}) > 0$ with probability at least $1 - \exp(-cT\sigma^4 + \mathcal{O}(k |\supp|)) - \exp(-\frac{\sigma^2 T}{128 \budget^2} + \mathcal{O}(k |\supp|)) $ for some constant $c > 0$ where $\Lambda_{\min}$ denotes the minimum eigenvalue.  
\end{lemma}

Next, we show that the mutual incoherence condition also holds in the finite-sample regime with high probability.
\begin{lemma}[Mutual incoherence in sample]
	\label{lemma:mutual_incoherence}
	If $\|\res{\bH}{\supp^c\supp} \res{\bH}{\supp \supp}^{-1}\|_{\BIO} \leq 1 - \alpha$ for $\alpha \in (0, 1]$ then $\|\res{\bHhat}{\supp^c \supp}\res{\bHhat}{\supp \supp}^{-1}\|_{\BIO} \leq 1 - \frac{\alpha}{2}$ with probability at least $ 1 - \mathcal{O}( \exp(\frac{-K T}{ k^5  |\supp|^3} + \log k |\supp^c| + \log k |\supp| )) $ for some $K > 0$. 
\end{lemma}

Now we can use Lemma \ref{lem:min eigenvalue} and \ref{lemma:mutual_incoherence} to prove our main result. We can rewrite equation \eqref{eq:support0} as, 
\begin{align}
\label{eq:support1}
\begin{split}
\res{\mb{W}_{i\cdot} - \mb{W}_{i\cdot}^*}{\supp \cdot}  =& {\res{\hE{\mb{x}_{\minus i} {\mb{x}_{\minus i}}^\T}}{\supp \supp}}^{\minus 1}  \res{\hE{\mb{x}_{\minus i} {\mb{e}_{\minus i}}^\T}}{\supp \supp} \res{\mb{W}_{i.}^*}{\supp\cdot} +  {\res{\hE{\mb{x}_{\minus i} {\mb{x}_{\minus i}}^\T}}{\supp \supp}}^{\minus 1} \hE{\res{\mb{x}_{\minus i}}{\supp\cdot} {\mb{e}_i}^\T} \\
 &- \frac{\lambda}{2}  {\res{\hE{\mb{x}_{\minus i} {\mb{x}_{\minus i}}^\T}}{\supp \supp}}^{\minus 1} \res{\mb{Z}_{i\cdot}}{\supp\cdot}
\end{split}
\end{align}
This is possible because $\lambda_{\min}(\res{\hE{\mb{x}_{\minus i} {\mb{x}_{\minus i}}^\T}}{\supp \supp}) > 0$ from Lemma \ref{lem:min eigenvalue}. Using equation \eqref{eq:support1}, we can write equation \eqref{eq:nonsupport} as,
\begin{align}
\label{eq:nonsupport1}
\begin{split}
 & \res{\hE{\mb{x}_{\minus i} {\mb{x}_{\minus i}}^\T}}{\supp^c \supp} \big(  {\res{\hE{\mb{x}_{\minus i} {\mb{x}_{\minus i}}^\T}}{\supp \supp}}^{\minus 1}  \res{\hE{\mb{x}_{\minus i} {\mb{e}_{\minus i}}^\T}}{\supp \supp} \res{\mb{W}_{i.}^*}{\supp\cdot}  
 +   {\res{\hE{\mb{x}_{\minus i} {\mb{x}_{\minus i}}^\T}}{\supp \supp}}^{\minus 1} \hE{\res{\mb{x}_{\minus i}}{\supp\cdot} {\mb{e}_i}^\T} \\
& - \frac{\lambda}{2}  {\res{\hE{\mb{x}_{\minus i} {\mb{x}_{\minus i}}^\T}}{\supp \supp}}^{\minus 1} \res{\mb{Z}_{i\cdot}}{\supp\cdot} \big) - \hE{ \res{\mb{x}_{\minus i} {\mb{e}_{\minus i}}^\T}{\supp^c \supp} \res{\mb{W}_{i.}^*}{\supp\cdot}} 
-  \hE{\res{\mb{x}_{\minus i}}{\supp^c\cdot} {\mb{e}_i}^\T} + \frac{\lambda}{2} \res{\mb{Z}_{i\cdot}}{\supp^c\cdot} = \mb{0}
\end{split}
\end{align}

Let $\mb{M} = \res{\hE{\mb{x}_{\minus i} {\mb{x}_{\minus i}}^\T}}{\supp^c \supp} \big(  {\res{\hE{\mb{x}_{\minus i} {\mb{x}_{\minus i}}^\T}}{\supp \supp}}^{\minus 1}$, then

\begin{align*}
	\frac{\lambda}{2} \res{\mb{Z}_{i\cdot}}{\supp^c\cdot} =& - \mb{M} \res{\hE{\mb{x}_{\minus i} {\mb{e}_{\minus i}}^\T}}{\supp \supp} \res{\mb{W}_{i.}^*}{\supp\cdot} -  \mb{M}  \hE{\res{\mb{x}_{\minus i}}{\supp\cdot} {\mb{e}_i}^\T}  
	+  \frac{\lambda}{2} \mb{M} \res{\mb{Z}_{i\cdot}}{\supp\cdot}  + \hE{\res{\mb{x}_{\minus i} {\mb{e}_{\minus i}}^\T}{\supp^c \supp}\\ &\res{\mb{W}_{i.}^*}{\supp\cdot} }+ \hE{\res{\mb{x}_{\minus i}}{\supp^c\cdot} {\mb{e}_i}^\T} 
\end{align*}
By taking the $\BIF$-norm on both sides and using the norm triangle inequality,
\begin{align*}
	\frac{\lambda}{2} \| \res{\mb{Z}_{i\cdot}}{\supp^c\cdot} \|_{\BIF} \leq& \| \mb{M} \res{\hE{\mb{x}_{\minus i} {\mb{e}_{\minus i}}^\T}}{\supp \supp} \res{\mb{W}_{i.}^*}{\supp\cdot}  - \mb{M}  \hE{\res{\mb{x}_{\minus i}}{\supp\cdot} {\mb{e}_i}^\T} 
	+  \frac{\lambda}{2} \mb{M} \res{\mb{Z}_{i\cdot}}{\supp\cdot} \|_{\BIF}  \\
	&+ \| \hE{\res{\mb{x}_{\minus i} {\mb{e}_{\minus i}}^\T}{\supp^c \supp}} \res{\mb{W}_{i.}^*}{\supp\cdot} \|_{\BIF} + \| \hE{\res{\mb{x}_{\minus i}}{\supp^c\cdot} {\mb{e}_i}^\T} \|_{\BIF} 
\end{align*}
Using the inequality $\| A B \|_{\BIF} \leq \| A \|_{\BIO} \|B\|_{\infty, 2}$ form Lemma \ref{lem:norminequalities}, we get
\begin{align*}
	\frac{\lambda}{2} \| \res{\mb{Z}_{i\cdot}}{\supp^c\cdot} \|_{\BIF} \leq& \| \mb{M} \|_{\BIO} \big( \| \res{\hE{\mb{x}_{\minus i} {\mb{e}_{\minus i}}^\T}}{\supp \supp} \res{\mb{W}_{i.}^*}{\supp\cdot}  -   \hE{\res{\mb{x}_{\minus i}}{\supp\cdot} {\mb{e}_i}^\T} 
	+  \frac{\lambda}{2} \res{\mb{Z}_{i\cdot}}{\supp\cdot} \|_{\infty,2} \big)  \\
	&+ \| \hE{\res{\mb{x}_{\minus i} {\mb{e}_{\minus i}}^\T}{\supp^c \supp} }\res{\mb{W}_{i.}^*}{\supp\cdot} \|_{\BIF}  + \| \hE{\res{\mb{x}_{\minus i}}{\supp^c\cdot} {\mb{e}_i}^\T }\|_{\BIF} 
\end{align*}
Again using the norm triangle inequality,
\begin{align}
	\label{eq:mainequations}
	\begin{split}
	\frac{\lambda}{2} \| \res{\mb{Z}_{i\cdot}}{\supp^c\cdot} \|_{\BIF} \leq& \| \mb{M} \|_{\BIO} \big( \| \res{\hE{\mb{x}_{\minus i} {\mb{e}_{\minus i}}^\T}}{\supp \supp} \res{\mb{W}_{i.}^*}{\supp\cdot} \|_{\infty, 2} 
	 + \|   \hE{\res{\mb{x}_{\minus i}}{\supp\cdot} {\mb{e}_i}^\T} \|_{\infty, 2} 
	+  \frac{\lambda}{2} \| \res{\mb{Z}_{i\cdot}}{\supp\cdot} \|_{\infty,2} \big)  \\
	&+ \| \hE{\res{\mb{x}_{\minus i} {\mb{e}_{\minus i}}^\T}{\supp^c \supp}} \res{\mb{W}_{i.}^*}{\supp\cdot} \|_{\BIF}  + \| \hE{\res{\mb{x}_{\minus i}}{\supp^c\cdot} {\mb{e}_i}^\T} \|_{\BIF} 
\end{split}
\end{align}
Next, we provide some technical lemmas (detailed proofs in Appendix \ref{appendix:proof of theorems and lemmas}) to bound all the terms in right hand side of equation \eqref{eq:mainequations}.
\begin{lemma}[Bound on \emph{$\| \res{\hE{\mb{x}_{\minus i} {\mb{e}_{\minus i}}^\T}}{\supp \supp} \res{\mb{W}_{i.}^*}{\supp\cdot} \|_{\infty, 2}$}]
	\label{lem: xew support}
	For some $\epsilon_1 > 0, 0 < \epsilon_2 < 8$ and $\epsilon_3 < 8 \frac{\sqrt{|\supp|} \max_{ij} |\mb{W}_{ij}^*|}{\min_{ij}|\mb{W}_{ij}^*|}$,
	\begin{align}
	\begin{split}
	&\prob\big(\| \res{\hE{\mb{x}_{\minus i} {\mb{e}_{\minus i}}^\T   }}{\supp \supp} \res{\mb{W}_{i.}^*}{\supp\cdot} \|_{\infty, 2} >  \epsilon_1 + \sigma^2 (1 + \epsilon_2 +\epsilon_3 ) \max_{i \in \ind(l), l \in \supp} \sqrt{\sum_{j=1}^{k} |\mb{W}_{ij}^*|^2} \big) \leq \\
	&\exp(-\frac{\epsilon_1^2 T}{2k \sigma^2 \budget^2 \max_{j}\sum_{k=1}^{|\supp|} {\mb{W}_{kj}^*}^2 } + \log(2 k^2 |\supp|) ) + k^2 |\supp| \exp(-\frac{T\epsilon_2^2}{64}) + \sum_{\substack{i \in \ind(l)\\ l \in \supp}} \sum_{j=1}^{k} \exp(-\frac{T\epsilon_3^2 |\mb{W}_{ij}^*|}{64 |\sqrt{\sum_{\substack{k=1 \\ k \ne i}}^{|\supp|}  {\mb{W}_{kj}^*}^2}|})
	\end{split}
	\end{align}
\end{lemma}

\begin{lemma}[Bound on \emph{$\|   \hE{\res{\mb{x}_{\minus i}}{\supp\cdot} {\mb{e}_i}^\T} \|_{\infty, 2} $}]
	\label{lem:bound sup xe}
	For some $\epsilon_4 > 0$ and $\epsilon_5 < 8 \sqrt{k} \sigma^2$,
	\begin{align}
	\begin{split}
	\prob(\| \hE{\res{\mb{x}_{\minus i}}{\supp\cdot} {\mb{e}_i}^\T}\|_{\infty, 2} \geq \epsilon_4 + \epsilon_5) \leq 
	\exp(- \frac{\epsilon_4^2 T}{2 k \sigma^2 \budget^2 } + \log(2k^2|\supp|)) + \exp(- \frac{\epsilon_5^2 T}{64 \sqrt{k} \sigma^2 } + \log(2k^2|\supp|))
	\end{split}
	\end{align}
\end{lemma}

\begin{lemma}[Bound on \emph{$\| \hE{\res{\mb{x}_{\minus i} {\mb{e}_{\minus i}}^\T}{\supp^c \supp}} \res{\mb{W}_{i.}^*}{\supp\cdot} \|_{\BIF}$}]
	\label{lem:bound x-ie-i}
	For some $\epsilon_6 > 0, 0 <\epsilon_7 < 8$ and $ \epsilon_8 <  8 \frac{\sqrt{|\supp^c|} \max_{ij} |\mb{W}_{ij}^*|}{\min_{ij}|\mb{W}_{ij}^*|}$,
	\begin{align}
	\begin{split}
	&\prob\big(\| \res{\hE{\mb{x}_{\minus i} {\mb{e}_{\minus i}}^\T   }}{\supp^c \supp} \res{\mb{W}_{i.}^*}{\supp\cdot} \|_{\BIF} > \epsilon_6 + \sigma^2 (1 + \epsilon_7 + \epsilon_8 ) \max_{l \in \supp^c} \sqrt{\sum_{i \in \ind(l)}\sum_{j=1}^{k} |\mb{W}_{ij}^*|^2} \big) \\
	&\leq\exp(-\frac{\epsilon_6^2 T}{2k^2 \sigma^2 \budget^2 \max_{j}\sum_{k=1}^{|\supp^c|} {\mb{W}_{kj}^*}^2 } + \log(2 k^2 |\supp^c|) ) +
	 k^2 |\supp^c| \exp(-\frac{T\epsilon_7^2}{64}) +\\
	& \sum_{i \in \ind(l), l \in \supp^c} \sum_{j=1}^{k} \exp(-\frac{T\epsilon_8^2 |\mb{W}_{ij}^*|}{64 |\sqrt{\sum_{\substack{k=1 \\ k \ne i}}^{|\supp^c|}  {\mb{W}_{kj}^*}^2}|})
	\end{split}
	\end{align}
\end{lemma}

\begin{lemma}[Bound on \emph{$\| \hE{\res{\mb{x}_{\minus i}}{\supp^c\cdot} {\mb{e}_i}^\T} \|_{\BIF} $}]
	\label{lem:bound x-ie_i}
	For some $\epsilon_9 > 0$ and $\epsilon_{10} < 8 k \sigma^2$,
	\begin{align}
	\begin{split}
	&\prob(\| \hE{\res{\mb{x}_{\minus i}}{\supp^c\cdot} {\mb{e}_i}^\T}\|_{\BIF} \geq \epsilon_9 + \epsilon_{10}) \leq \exp(- \frac{\epsilon_9^2 T}{2 k^2 \sigma^2 \budget^2 } + \log(2k^2|\supp^c|)) + \exp(- \frac{\epsilon_{10}^2 T}{64 k \sigma^2 } + \log(2k^2|\supp^c|))
	\end{split}
	\end{align}
\end{lemma}
%
%
%
Recall that,
\begin{align}
\begin{split}
	\frac{\lambda}{2} \| \res{\mb{Z}_{i\cdot}}{\supp^c\cdot} \|_{\BIF} \leq& \| \mb{M} \|_{\BIO} \big( \| \res{\hE{\mb{x}_{\minus i} {\mb{e}_{\minus i}}^\T}}{\supp \supp} \res{\mb{W}_{i.}^*}{\supp\cdot} \|_{\infty, 2} 
	+ \|   \hE{\res{\mb{x}_{\minus i}}{\supp\cdot} {\mb{e}_i}^\T} \|_{\infty, 2} 
	+  \frac{\lambda}{2} \| \res{\mb{Z}_{i\cdot}}{\supp\cdot} \|_{\infty,2} \big)  \\
	&+ \| \hE{\res{\mb{x}_{\minus i}^t {\mb{e}_{\minus i}^t}^\T}{\supp^c \supp}} \res{\mb{W}_{i.}^*}{\supp\cdot} \|_{\BIF}  	+ \| \hE{\res{\mb{x}_{\minus i}^t}{\supp^c\cdot} {\mb{e}_i^t}^\T} \|_{\BIF} 
\end{split}
\end{align}

We already showed that mutual incoherence holds in the finite-sample regime, i.e., $\| \mb{M} \|_{\BIO} \leq 1 - \alpha$ for some $0 <\alpha < 1$. Also note that $\| \res{\mb{Z}_{i\cdot}}{\supp\cdot} \|_{\infty,2} \leq 1$. It follows that,
\begin{align*}
& \frac{\lambda}{2} \| \res{\mb{Z}_{i\cdot}}{\supp^c\cdot} \|_{\BIF} \leq (1 - \alpha) \big( \| \res{\hE{\mb{x}_{\minus i} {\mb{e}_{\minus i}}^\T}}{\supp \supp} \res{\mb{W}_{i.}^*}{\supp\cdot} \|_{\infty, 2}  + \|   \hE{\res{\mb{x}_{\minus i}}{\supp\cdot} {\mb{e}_i}^\T} \|_{\infty, 2} 
+  \frac{\lambda}{2}  \big)  \\
&+ \| \hE{\res{\mb{x}_{\minus i} {\mb{e}_{\minus i}}^\T}{\supp^c \supp}} \res{\mb{W}_{i.}^*}{\supp\cdot} \|_{\BIF}  	+\| \hE{\res{\mb{x}_{\minus i}}{\supp^c\cdot} {\mb{e}_i}^\T} \|_{\BIF} 
\end{align*}
Using bounds from Lemmas \ref{lem: xew support}, \ref{lem:bound sup xe}, \ref{lem:bound x-ie-i},  and \ref{lem:bound x-ie_i}, we get
\begin{align*}
\frac{\lambda}{2} \| \res{\mb{Z}_{i\cdot}}{\supp^c\cdot} \|_{\BIF} \leq& (1 - \alpha) \big( \epsilon_1 +  \sigma^2(1 + \epsilon_2 + \epsilon_3) \max_{i \in \ind(l), l \in \supp} \sqrt{\sum_{j=1}^k |\mb{W}_{ij}^2|}  + \epsilon_4 + \epsilon_5 +  \frac{\lambda}{2}  \big)  + \epsilon_6 + \sigma^2\\
& (1 + \epsilon_7 + \epsilon_8 ) \max_{l \in \supp^c} \sqrt{\sum_{i \in \ind(l)}\sum_{j=1}^{k} |\mb{W}_{ij}^*|^2}   + \epsilon_9 + \epsilon_{10}
\end{align*}
where,
\begin{align*}
&\epsilon_1 > 0, \, 0 < \epsilon_2 < 8, \, \epsilon_3 < 8 \frac{\sqrt{|\supp|} \max_{ij} |\mb{W}_{ij}^*|}{\min_{ij}|\mb{W}_{ij}^*|} \epsilon_4 > 0, \, \epsilon_5 < 8 \sqrt{k} \sigma^2, \, \epsilon_6 > 0, \, 0 < \epsilon_7 < 8, \\
&\epsilon_8 <  8 \frac{\sqrt{|\supp^c|} \max_{ij} |\mb{W}_{ij}^*|}{\min_{ij}|\mb{W}_{ij}^*|}, \, \epsilon_9 > 0, \, \epsilon_{10} < 2k\sigma^2
\end{align*}
If these conditions hold, then
\begin{align*}
	\frac{\lambda}{2} \| \res{\mb{Z}_{i\cdot}}{\supp^c\cdot} \|_{\BIF} \leq& (1 - \alpha) \big( \epsilon_1 + \sigma^2(1 + \epsilon_2 + \epsilon_3) \sqrt{k} \max_{ij}  |\mb{W}_{ij}|  + \epsilon_4 + \epsilon_5 +  \frac{\lambda}{2}  \big)  + \epsilon_6 + \sigma^2\\
	& (1 + \epsilon_7 + \epsilon_8 ) k \max_{ij} |\mb{W}_{ij}^*|   	+ \epsilon_9 + \epsilon_{10} 
\end{align*}
After rearranging the terms, we get,
\begin{align}
	\label{eq:zsc final} 
	\begin{split}
	&\frac{\lambda}{2} \| \res{\mb{Z}_{i\cdot}}{\supp^c\cdot} \|_{\BIF} \leq (1 - \alpha) \epsilon_1 + 
	(1 - \alpha)  \sigma^2 \sqrt{k} \max_{ij}  |\mb{W}_{ij}| + (1 - \alpha) \sigma^2  \epsilon_2 \sqrt{k} \max_{ij}  |\mb{W}_{ij}| + (1 - \alpha) \\ 
	&\sigma^2 \epsilon_3 \sqrt{k} \max_{ij}  |\mb{W}_{ij}|  + 
	(1-\alpha)\epsilon_4 + (1 - \alpha)\epsilon_5 + (1-\alpha)\frac{\lambda}{2}  +  \epsilon_6 + \sigma^2 k \max_{ij} |\mb{W}_{ij}^*| + \epsilon_7 \sigma^2 k \\ 
	&\max_{ij} |\mb{W}_{ij}^*| + \sigma^2 \epsilon_8  k \max_{ij} |\mb{W}_{ij}^*|   	+ \epsilon_9 + \epsilon_{10} 
	 \end{split}
\end{align}
\paragraph{Choice of $\lambda$.}

To keep the RHS of equation \eqref{eq:zsc final} less than $\frac{\lambda}{2}$, we need to set $\epsilon_1 < \frac{\alpha \lambda}{24(1-\alpha)}, \epsilon_2 < \frac{\alpha \lambda}{24 (1-\alpha)\sigma^2 \sqrt{k} \max_{ij} |W_{ij}|}, \epsilon_3 < \frac{\alpha \lambda}{24 (1 - \alpha) \sigma^2 \sqrt{k} \max_{ij} |W_{ij}|}, \epsilon_4 < \frac{\alpha \lambda}{24 (1-\alpha)}, \epsilon_5 < \frac{\alpha \lambda}{24 (1-\alpha)}, \epsilon_6 < \frac{\alpha \lambda}{24}, \epsilon_7 < \frac{\alpha \lambda}{24 \sigma^2 k \max_{ij} |W_{ij}|}, \epsilon_8 <  \frac{\alpha \lambda}{24 \sigma^2 k \max_{ij} |W_{ij}|}, \epsilon_9 < \frac{\alpha \lambda}{24}$ and  $\epsilon_{10} < \frac{\alpha \lambda}{24}$. We also want to make sure that claim in Lemma \ref{lem: xew support}, \ref{lem:bound sup xe}, \ref{lem:bound x-ie-i} and \ref{lem:bound x-ie_i} hold with high probability. This can be achieved by keeping a $\lambda$ such that 
$ \lambda > \max( 24\sqrt{2} \frac{1-\alpha}{\alpha} \sigma \budget W_{\max} \sqrt{\frac{k |\supp| \log( 2k^2 |\supp|)}{T}},  192 \frac{1-\alpha}{\alpha} \sigma^2  W_{\max} \sqrt{\frac{k \log(k^2 |\supp|)}{T}}, 192 \frac{1-\alpha}{\alpha} \sigma^2 \sqrt{k} W_{\max} \\ \sqrt{\frac{W_{\max} |\supp| \log(|\supp|k)}{TW_{\min}}},  192 \frac{1-\alpha}{\alpha} k^{\frac{1}{4}} \sigma \sqrt{\frac{\log (2k^2|\supp|)}{T}},  \frac{24\sqrt{2}}{\alpha} k \sigma \budget W_{\max} \sqrt{\frac{|\supp^c \log(2k^2|\supp^c|)|}{T}},  \frac{192}{\alpha} \sigma^2 k W_{\max} \\ \sqrt{\frac{\log(k^2 |\supp^c|)}{T}}, \frac{192}{\alpha} \sigma^2 k W_{\max}  \sqrt{\frac{W_{\max} |\supp^c|^{\frac{1}{2}} \log(|\supp^c|k) }{T}},   \frac{24\sqrt{2}}{\alpha} k \sigma \budget \sqrt{\frac{\log (2k^2 |\supp^c|)}{T}} , \frac{192}{\alpha} \sigma \sqrt{\frac{k \log(2k^2|\supp^c|)}{T}},\\ \frac{24 (1-\alpha) \sigma^2 \sqrt{k} \max_{ij} |W_{ij}| }{\alpha}, \frac{24 \sigma^2 k \max_{ij} |W_{ij}| }{\alpha} ) $.

This particular choice of $\lambda$ implies that $\frac{\lambda}{2} \| \res{\mb{Z}_{i\cdot}}{\supp^c\cdot} \|_{\BIF}  < \frac{\lambda}{2}$
with high probability which in turn ensures that $W_{ij}$ are zero for all $j \in \supp^c$ with high probability. Now,
\begin{align}
\begin{split}
	\res{\mb{W}_{i\cdot} - \mb{W}_{i\cdot}^*}{\supp \cdot}  =& {\res{\hE{\mb{x}_{\minus i} {\mb{x}_{\minus i}}^\T}}{\supp \supp}}^{\minus 1}  \res{\hE{\mb{x}_{\minus i} {\mb{e}_{\minus i}}^\T}}{\supp \supp} \res{\mb{W}_{i.}^*}{\supp\cdot} +  {\res{\hE{\mb{x}_{\minus i} {\mb{x}_{\minus i}}^\T}}{\supp \supp}}^{\minus 1} \hE{\res{\mb{x}_{\minus i}}{\supp\cdot} {\mb{e}_i}^\T} \\
	&- \frac{\lambda}{2}  {\res{\hE{\mb{x}_{\minus i} {\mb{x}_{\minus i}}^\T}}{\supp \supp}}^{\minus 1} \res{\mb{Z}_{i\cdot}}{\supp\cdot} 
\end{split}
\end{align}
By taking the $\BIF$-norm on both sides,
\begin{align*}
	\| \res{\mb{W}_{i\cdot} - \mb{W}_{i\cdot}^*}{\supp \cdot} \|_{\BIF}  =& \| {\res{\hE{\mb{x}_{\minus i} {\mb{x}_{\minus i}}^\T}}{\supp \supp}}^{\minus 1}  \res{\hE{\mb{x}_{\minus i} {\mb{e}_{\minus i}}^\T}}{\supp \supp} \res{\mb{W}_{i.}^*}{\supp\cdot} +  {\res{\hE{\mb{x}_{\minus i} {\mb{x}_{\minus i}}^\T}}{\supp \supp}}^{\minus 1} \\
	& \hE{\res{\mb{x}_{\minus i}}{\supp\cdot} {\mb{e}_i}^\T}  \frac{\lambda}{2}  {\res{\hE{\mb{x}_{\minus i} {\mb{x}_{\minus i}}^\T}}{\supp \supp}}^{\minus 1} \res{\mb{Z}_{i\cdot}}{\supp\cdot} \|_{\BIF} 
\end{align*}
Using the norm triangle inequality,
\begin{align*}
	\| \res{\mb{W}_{i\cdot} - \mb{W}_{i\cdot}^*}{\supp \cdot} \|_{\BIF} &\leq  \| {\res{\hE{\mb{x}_{\minus i} {\mb{x}_{\minus i}}^\T}}{\supp \supp}}^{\minus 1}  \res{\hE{\mb{x}_{\minus i} {\mb{e}_{\minus i}}^\T}}{\supp \supp} \res{\mb{W}_{i.}^*}{\supp\cdot} \|_{\BIF} + \\
	& \| {\res{\hE{\mb{x}_{\minus i} {\mb{x}_{\minus i}}^\T}}{\supp \supp}}^{\minus 1} \hE{\res{\mb{x}_{\minus i}}{\supp\cdot} {\mb{e}_i}^\T} \|_{\BIF}  + \| \frac{\lambda}{2}  {\res{\hE{\mb{x}_{\minus i} {\mb{x}_{\minus i}}^\T}}{\supp \supp}}^{\minus 1} \res{\mb{Z}_{i\cdot}}{\supp\cdot} \|_{\BIF}
\end{align*}
Using the inequality $\| A B \|_{\BIF} \leq \| A \|_{\BIO} \| B \|_{\infty, 2} $ from Lemma \ref{lem:norminequalities}, we get
\begin{align*}
	&\| \res{\mb{W}_{i\cdot} - \mb{W}_{i\cdot}^*}{\supp \cdot} \|_{\BIF} \leq  \| {\res{\hE{\mb{x}_{\minus i} {\mb{x}_{\minus i}}^\T}}{\supp \supp}}^{\minus 1} \|_{\BIO} \| \res{\hE{\mb{x}_{\minus i} {\mb{e}_{\minus i}}^\T}}{\supp \supp} \res{\mb{W}_{i.}^*}{\supp\cdot} \|_{\infty, 2} + \\
	& \| {\res{\hE{\mb{x}_{\minus i} {\mb{x}_{\minus i}}^\T}}{\supp \supp}}^{\minus 1} \|_{\BIO} \| \hE{\res{\mb{x}_{\minus i}}{\supp\cdot} {\mb{e}_i}^\T} \|_{\infty, 2}  + 
	 \| \frac{\lambda}{2}  {\res{\hE{\mb{x}_{\minus i} {\mb{x}_{\minus i}}^\T}}{\supp \supp}}^{\minus 1} \|_{\BIF} \| \res{\mb{Z}_{i\cdot}}{\supp\cdot} \|_{\infty, 2}
\end{align*}
Using the inequality $\| A \|_{\BIO} \leq k \| A \|_{\infty, \infty}$, where $k$ is the maximum number of rows in a block of $A$, we obtain
\begin{align*}
	&\| \res{\mb{W}_{i\cdot} - \mb{W}_{i\cdot}^*}{\supp \cdot} \|_{\BIF} \leq  k \| {\res{\hE{\mb{x}_{\minus i} {\mb{x}_{\minus i}}^\T}}{\supp \supp}}^{\minus 1} \|_{\infty, \infty} \| \res{\hE{\mb{x}_{\minus i} {\mb{e}_{\minus i}}^\T}}{\supp \supp} \res{\mb{W}_{i.}^*}{\supp\cdot} \|_{\infty, 2} + \\  
	& k \| {\res{\hE{\mb{x}_{\minus i} {\mb{x}_{\minus i}}^\T}}{\supp \supp}}^{\minus 1} \|_{\infty, \infty} \| \hE{\res{\mb{x}_{\minus i}}{\supp\cdot} {\mb{e}_i}^\T} \|_{\infty, 2}  +  \frac{\lambda}{2}  k \| {\res{\hE{\mb{x}_{\minus i} {\mb{x}_{\minus i}}^\T}}{\supp \supp}}^{\minus 1} \|_{\infty, \infty} 
\end{align*}
Since $\| A \|_{\infty, \infty} \leq \sqrt{p} \| A \|_{2, 2} $ for $A \in \real^{p \times p}$,
\begin{align*}
	& \| \res{\mb{W}_{i\cdot} - \mb{W}_{i\cdot}^*}{\supp \cdot} \|_{\BIF} \leq  k \sqrt{k|\supp|} \| {\res{\hE{\mb{x}_{\minus i} {\mb{x}_{\minus i}}^\T}}{\supp \supp}}^{\minus 1} \|_{2, 2} \| \res{\hE{\mb{x}_{\minus i} {\mb{e}_{\minus i}}^\T}}{\supp \supp} \res{\mb{W}_{i.}^*}{\supp\cdot} \|_{\infty, 2} \\
	&+ k \sqrt{k|\supp|} \| {\res{\hE{\mb{x}_{\minus i} {\mb{x}_{\minus i}}^\T}}{\supp \supp}}^{\minus 1} \|_{2, 2} \| \hE{\res{\mb{x}_{\minus i}}{\supp\cdot} {\mb{e}_i}^\T} \|_{\infty, 2}  + \frac{\lambda}{2}  k \sqrt{k|\supp|} \| {\res{\hE{\mb{x}_{\minus i} {\mb{x}_{\minus i}}^\T}}{\supp \supp}}^{\minus 1} \|_{2, 2} 
\end{align*}
Substituting for $\|{\res{\hE{\mb{x}_{\minus i} {\mb{x}_{\minus i}}^\T}}{\supp \supp}}^{\minus 1}  \|_{2,2}$,
\begin{align*} 
	\| \res{\mb{W}_{i\cdot} - \mb{W}_{i\cdot}^*}{\supp \cdot} \|_{\BIF} \leq&  k \sqrt{k|\supp|} \frac{2}{C_{\min}} \| \res{\hE{\mb{x}_{\minus i} {\mb{e}_{\minus i}}^\T}}{\supp \supp} \res{\mb{W}_{i.}^*}{\supp\cdot} \|_{\infty, 2} + k \\
	& \sqrt{k|\supp|} \frac{2}{C_{\min}} \| \hE{\res{\mb{x}_{\minus i}}{\supp\cdot} {\mb{e}_i}^\T} \|_{\infty, 2}  +   \frac{\lambda}{2}  k \sqrt{k|\supp|}\frac{2}{C_{\min}} 
\end{align*}
Using results from Lemma \ref{lem: xew support} and \ref{lem:bound sup xe}, with high probability,
\begin{align*} 
\| \res{\mb{W}_{i\cdot} - \mb{W}_{i\cdot}^*}{\supp \cdot} \|_{\BIF} &\leq  k \sqrt{k|\supp|} \frac{2}{C_{\min}} (\epsilon_1 + \sigma^2 (1 + \epsilon_2 + \epsilon_3) + k \sqrt{k|\supp|}\\ 
&\frac{2}{C_{\min}} 
(\epsilon_4 + \epsilon_5)  +   \frac{\lambda}{2}  k \sqrt{k|\supp|}\frac{2}{C_{\min}} 
\end{align*}
Substituting bounds on $\epsilon_1, \epsilon_2, \epsilon_3, \epsilon_4$ and $\epsilon_5$, we get
\begin{align*} 
&\| \res{\mb{W}_{i\cdot} - \mb{W}_{i\cdot}^*}{\supp \cdot} \|_{\BIF} \leq  k \sqrt{k|\supp|} \frac{2}{C_{\min}} (\frac{\alpha \lambda}{24(1-\alpha)} + \sigma^2 (1 + \frac{\alpha \lambda}{24 (1-\alpha)\sigma^2 \sqrt{k} W_{\max}} +\\
& \frac{\alpha \lambda}{24 (1 - \alpha) \sigma^2 \sqrt{k} \max_{ij} W_{\max}}) \sqrt{k} W_{\max}) + k \sqrt{k|\supp|}  \frac{2}{C_{\min}} (\frac{\alpha \lambda}{24 (1-\alpha)} + \frac{\alpha \lambda}{24 (1-\alpha)})  + \\
&\frac{\lambda}{2}  k \sqrt{k|\supp|}\frac{2}{C_{\min}} 
\end{align*}
Let 
\begin{align*} 
&\delta(k, |\supp|, C_{\min}, \alpha, \lambda, \sigma, W_{\max}) =  k \sqrt{k|\supp|} \frac{2}{C_{\min}} (\frac{\alpha \lambda}{24(1-\alpha)} + \sigma^2 (1 + \frac{\alpha \lambda}{24 (1-\alpha)\sigma^2 \sqrt{k} W_{\max}} +\\
& \frac{\alpha \lambda}{24 (1 - \alpha) \sigma^2 \sqrt{k} \max_{ij} W_{\max}}) \sqrt{k} W_{\max}) + k \sqrt{k|\supp|}  \frac{2}{C_{\min}} (\frac{\alpha \lambda}{24 (1-\alpha)} + \frac{\alpha \lambda}{24 (1-\alpha)})  + \\
&\frac{\lambda}{2}  k \sqrt{k|\supp|}\frac{2}{C_{\min}} 
\end{align*}
Then,
\begin{align*} 
&\| \res{\mb{W}_{i\cdot} - \mb{W}_{i\cdot}^*}{\supp \cdot} \|_{\BIF} \leq  \delta(k, |\supp|, C_{\min}, \alpha, \lambda, \sigma, W_{\max})
\end{align*}

Now, we will characterize $\nash_{\epsilon}(\game)$ by $\mb{W}_{i\cdot}$. In particular, we define
\begin{align*}
	\label{eq:infer eps nash}
	\nash_{\epsilon}(\game) =& \{ \mb{x}^* \in \bigtimes_{i \in \nodes} \actions_i \mid x_i^* = \sum_{j \in \supp} W_{ij} x_j^* , \forall i \in \seq{n}  \}
\end{align*}
We explicitly compute the payoffs to prove that equation \eqref{eq:infer eps nash} indeed recovers $\nash_{\epsilon}(\game)$, i.e., for all $\mb{x}^* \in \nash_{\epsilon}(\game)$
\begin{align*}
&u_i(x_i^*, x_{\minus i}^*) = - \| \sum_{j \in \supp} W_{ij} x_j^* - \sum_{j \in \supp} W_{ij}^* x_j^* \|_2 \geq - \sum_{j \in \supp} \| (W_{ij} - W_{ij}^*) x_j^* \|_2 
\geq - \sum_{j \in \supp} \| W_{ij} - W_{ij}^* \|_{F} \| x_j \|_2 \\
&\geq - |\supp| \delta(k, |\supp|, C_{\min}, \alpha, \lambda, \sigma, W_{\max}) \budget
\end{align*}
Thus the set defined in equation \eqref{eq:infer eps nash} recovers $\epsilon$-PSNE for $\epsilon = |\supp| \delta(k, |\supp|, C_{\min}, \alpha, \lambda, \sigma, W_{\max}) \budget$.

Next, we show that if for each player $i \in \seq{n}$, if $\min_{j \in \supp} \|W_{ij}^*\|_{F} > 2 \delta(k, |\supp|, C_{\min}, \alpha, \lambda, \sigma, W_{\max})$ then we recover the exact structure of the graphical game. Note that if $\min_{j \in \supp} \|W_{ij}^*\|_{F} > 2 \delta(k, |\supp|, \\ C_{\min}, \alpha, \lambda, \sigma, W_{\max})$ then $\|W_{ij}^*\|_{F} > 0$ implies that  $\|W_{ij}\|_{F} > 0$. We have already shown that we do not recover any extra player in the set of in-neighbors $\supp$ and this added condition ensures that for all the players in $\supp$, $\|W_{ij}\|_{F} > 0$. Thus, we recover exact set of players in $\supp$ for each player $i \in \seq{n}$. We recover the exact graphical game by combining the results for all the players. 
\end{proof} 

\paragraph{Sample and Time Complexity.}
\label{sec:sample and time complexity} 
If we have $T > \mathcal{O}(k^5 |\supp|^3 \log (k |\supp^c||\supp|))$ and all other conditions mentioned in Theorem \ref{thm:eps_nash_equilibria}  are satisfied for every player then all our high probability statements are valid for every player $i$. Taking a union bound over $n$ players only adds a factor of $\log n$. Thus the sample complexity for our method is $\mathcal{O}(k^5 |\supp|^3 \log (k |\supp^c||\supp|))$. As for the time complexity, we can formulate the block-regularized multi-variate regression problem as a second order cone programing problem~\cite{obozinski2011support} which can be solved in polynomial time by interior point methods~\cite{boyd2004convex}. 

\paragraph{Concluding Remarks.}
There are two possible future directions for our work. First, it would be interesting to see if the exact Nash equilibria set can be recovered for  graphical games with quadratic payoffs using our method. Second, it would be interesting to extend our method to the continuous-action games with more general payoff functions.

\appendix 

\section{Proofs of Theorems and Lemmas}
\label{appendix:proof of theorems and lemmas}

\subsection{Proof of Lemma \ref{lem:min eigenvalue}}
\paragraph{Lemma \ref{lem:min eigenvalue}}[Positive minimum eigenvalue]
	\emph{$\Lambda_{\min}(\res{\bHhat}{\supp \supp}) > 0$ with probability at least $1 - \exp(-cT\sigma^4 + \mathcal{O}(k |\supp|)) - \exp(-\frac{\sigma^2 T}{128 \budget^2} + \mathcal{O}(k |\supp|)) $ for some constant $c > 0$ where $\Lambda_{\min}$ denotes the minimum eigenvalue.} 
\begin{proof}
	We prove the lemma in two steps. First, recall that $\bH = \frac{1}{T} \sum_{t=1}^T \big( {\mb{x}_{\minus i}^*}^t {{\mb{x}_{\minus i}^*}^t}^\T +  \sigma^2 \mb{I}\big)$ where $\mb{I}$ is identity matrix. Then,
	\begin{align}
	\begin{split}
	\Lambda_{\min}(\res{\bH}{\supp \supp}) &= \Lambda_{\min}(\frac{1}{T}\sum_{t=1}^T (\res{{\mb{x}_{\minus i}^*}^t {{\mb{x}_{\minus i}^*}^t}^\T}{\supp\supp})  + \sigma^2 \mb{I}) \\
	&\text{Using the inequality}\\
	& \text{$\Lambda_{\min}(A+B) \geq \Lambda_{\min}(A) + \Lambda_{\min}(B)$}\\
	&\geq \Lambda_{\min}(\frac{1}{T}\sum_{t=1}^T \res{{\mb{x}_{\minus i}^*}^t {{\mb{x}_{\minus i}^*}^t}^\T}{\supp\supp}) + \sigma^2 \\
	&\geq \sigma^2 > 0
	\end{split}
	\end{align}
	Last inequality follows by noting that $\frac{1}{T}\sum_{t=1}^T (\res{{\mb{x}_{\minus i}^*}^t {{\mb{x}_{\minus i}^*}^t}^\T}{\supp\supp})$ is a positive semi-definite matrix with non-negative eigenvalues. Next, we  prove that if $T > \mathcal{O}\mathcal{O}(\frac{1}{\sigma^2}\max(\budget^2, \frac{1}{\sigma^2}) k |\supp|)$, then $\Lambda_{\min}(\res{\bHhat}{\supp \supp}) > 0$ with high probability. 
	\begin{align*}
			\Lambda_{\min}(\res{\bHhat}{\supp\supp}) &= \min_{\|\mb{y}\|_2 = 1} \mb{y}^\T (\frac{1}{T}\sum_{t=1}^T \res{{\mb{x}^*}^t{{\mb{x}^*}^t}^\T}{\supp\supp} + \res{{\mb{x}^*}^t {\mb{e}^t}^\T}{\supp\supp} + \res{\mb{e}^t {{\mb{x}^*}^t}^\T}{\supp\supp} + \res{{\mb{e}}^t{\mb{e}^t}^\T}{\supp\supp}) \mb{y}\\
			&\geq \min_{\|\mb{y}\|_2 = 1} \mb{y}^\T \frac{1}{T}\sum_{t=1}^T \res{{\mb{x}^*}^t{{\mb{x}^*}^t}^\T}{\supp\supp} \mb{y}+  \min_{\|\mb{y}\|_2 = 1} \mb{y}^\T\frac{1}{T}\sum_{t=1}^T  (\res{{\mb{x}^*}^t {\mb{e}^t}^\T}{\supp\supp} + \res{\mb{e}^t {{\mb{x}^*}^t}^\T}{\supp\supp}) \mb{y}+  \\
			&\min_{\|\mb{y}\|_2 = 1} \mb{y}^\T \frac{1}{T}\sum_{t=1}^T \res{{\mb{e}}^t{\mb{e}^t}^\T}{\supp\supp} \mb{y} \\
			&\text{Noting that $\frac{1}{T}\sum_{t=1}^T (\res{{\mb{x}_{\minus i}^*}^t {{\mb{x}_{\minus i}^*}^t}^\T}{\supp\supp})$ is a positive semidefinite matrix}\\
			&\geq  \min_{\|\mb{y}\|_2 = 1}  \mb{y}^\T \frac{1}{T}\sum_{t=1}^T (\res{{\mb{x}^*}^t {\mb{e}^t}^\T}{\supp\supp} + \res{\mb{e}^t {{\mb{x}^*}^t}^\T}{\supp\supp}) \mb{y}+  \min_{\|\mb{y}\|_2 = 1} \mb{y}^\T \frac{1}{T}\sum_{t=1}^T \res{{\mb{e}}^t{\mb{e}^t}^\T}{\supp\supp} \mb{y}
	\end{align*}
	We define a random variable $R \triangleq = \frac{1}{T} \sum_{t=1}^T \mb{y}^\T(\res{{\mb{x}^*}^t {\mb{e}^t}^\T}{\supp\supp} + \res{\mb{e}^t {{\mb{x}^*}^t}^\T}{\supp\supp})\mb{y}$. Notice that $R$ is a sub-Gaussian random variable with mean $0$ and parameter $ \frac{4 \sum_{t=1}^T a_t^2 \sigma^2}{T^2}$, where $a_t = \mb{y}^\T \res{\mb{x}^*}{\supp.} \leq \budget$.
	Thus, 
	\begin{align*}
		\prob(R \leq -\epsilon) \leq \exp(- \frac{\epsilon^2}{2 \frac{4 \sum_{t=1}^T a_t^2 \sigma^2}{T^2} }) \leq \exp(- \frac{\epsilon^2}{2 \frac{4 \sum_{t=1}^T \budget^2\sigma^2}{T^2} }) = \exp(- \frac{T \epsilon^2}{8 \budget^2 \sigma^2 }) 
	\end{align*}
	Following \emph{$\epsilon$-nets} argument from \cite{vershynin2012close} and covariance matrix concentration for sub-Gaussian random variables, we can write
	\begin{align}
		\prob(\min_{\|\mb{y}\|_2 = 1} \mb{y}^\T \frac{1}{T}\sum_{t=1}^T \res{{\mb{e}}^t{\mb{e}^t}^\T}{\supp\supp} \mb{y} > \sigma^2 - \epsilon) \geq 1 - \exp(-c\epsilon^2T + \mathcal{O}(k |\supp|))
	\end{align}
	and 
	\begin{align}
	\prob(\min_{\|\mb{y}\|_2 = 1}  \mb{y}^\T \frac{1}{T}\sum_{t=1}^T (\res{{\mb{x}^*}^t {\mb{e}^t}^\T}{\supp\supp} + \res{\mb{e}^t {{\mb{x}^*}^t}^\T}{\supp\supp}) \mb{y} \geq -\epsilon) \leq  \exp(- \frac{T \epsilon^2}{8 \budget^2 \sigma^2 } + \mathcal{O}(k |\supp|)) 
	\end{align}
	Thus, choosing $\epsilon = \frac{\sigma^2}{4}$ and choosing $T = \mathcal{O}(\frac{1}{\sigma^2}\max(\budget^2, \frac{1}{\sigma^2}) k |\supp|)$, we get $\Lambda_{\min}(\res{\bHhat}{\supp\supp}) \geq \frac{\sigma^2}{2}$ with high probability.
\end{proof}

\subsection{Proof of Lemma \ref{lemma:mutual_incoherence}}

First, we prove a technical lemma that will be used in Lemma \ref{lemma:mutual_incoherence}.
\begin{lemma}
	\label{lem:control_T moved}
	For any $\delta > 0$, the following holds:
	\begin{align}
	\label{eq:Teq1}
	\begin{split}
	\prob(\| \res{\bHhat}{\supp^c \supp} - \res{\bH}{\supp^c \supp}\|_{\BIO} \geq \delta) \leq 2 D \exp(\frac{-\delta^2 T}{8  f(\sigma, \budget) k^4 |\supp|2} 
	+ \log k |\supp^c| + \log k |\supp| )
	\end{split}
	\end{align}
	\begin{align}
	\label{eq:Teq2}
	\begin{split}
	\prob(\| \res{\bHhat}{\supp \supp} - \res{\bH}{\supp \supp} \|_{\infty, \infty} \geq \delta) \leq 2 D \exp(\frac{-\delta^2 T}{8  f(\sigma, \budget) k^2 |\supp|^2  } 		+ 2 \log k |\supp|)
	\end{split}
	\end{align}
	\begin{align}
	\label{eq:Teq3}
	\begin{split}
	\prob(\|(\res{\bHhat}{\supp \supp})^{-1} - (\res{\bH}{\supp \supp})^{-1}\|_{\infty, \infty} \geq \delta) \leq 
	2D \exp(- \frac{\delta^2 C^4 T }{32 f(\sigma, \budget) k^3 |\supp|^3} + 2 \log k|\supp|) \\
	+ 2D \exp(- \frac{C^2 T }{32  f(\sigma, \budget)k^2 |\supp|^2} + 2 \log k |\supp|)
	\end{split}
	\end{align}
\end{lemma}
\begin{proof}
	Note that,
	\begin{align*}
	\begin{split}
	[ \res{\bHhat}{\supp^c \supp} - \res{\bH}{\supp^c \supp}]_{jk} &= Z_{jk} \\
	&= [\frac{1}{T} \sum_{t=1}^T \big( {\mb{x}_{\minus i}^*}^t {{\mb{e}_{\minus i}}^t}^\T + {\mb{e}_{\minus i}}^t {{\mb{x}_{\minus i}^*}^t}^\T + {{\mb{e}_{\minus i}}^t}{{\mb{e}_{\minus i}}^t}^\T \\
	&- \sigma^2 \mb{I} \big)]_{ij} \\
	&= \frac{1}{T}\sum_{t=1}^T Z_{jk}^t
	\end{split}
	\end{align*}
	
	We define three random variables $R_1 \triangleq [\frac{1}{T} \sum_{t=1}^T \big( {\mb{x}_{\minus i}^*}^t {{\mb{e}_{\minus i}}^t}^\T\big)]_{jk}, R_2 = [\frac{1}{T} \sum_{t=1}^T\big( {\mb{e}_{\minus i}}^t {{\mb{x}_{\minus i}^*}^t}^\T \big)]_{jk}$ and $R_3 \triangleq [\frac{1}{T} \sum_{t=1}^T \big({{\mb{e}_{\minus i}}^t}{{\mb{e}_{\minus i}}^t}^\T - \sigma^2 \mb{I} \big)]_{jk}$. We will provide a separate bound on these random variables. Note that $R_1$ is a sub-Gaussian random variable with $0$ mean and parameter $\frac{\sigma^2 \sum_{t=1}^T {{{\mb{x}_{\minus i}}_j^*}^t}^2}{T^2}$ and $R_2$ is a sub-Gaussian random variable with $0$ mean and parameter $\frac{\sigma^2 \sum_{t=1}^T {{{\mb{x}_{\minus i}}_k^*}^t}^2}{T^2}$. Thus,
		\begin{align*}
		\prob(|R_1| \geq \epsilon) \leq \exp(- \frac{\epsilon^2}{\frac{\sigma^2 \sum_{t=1}^T {{{\mb{x}_{\minus i}}_j^*}^t}^2}{T^2}}) \leq 2\exp(- \frac{T \epsilon^2}{ \budget^2 \sigma^2 }) 
		\end{align*}
		and
		\begin{align*}
		\prob(|R_2| \geq \epsilon) \leq \exp(- \frac{\epsilon^2}{\frac{\sigma^2 \sum_{t=1}^T {{{\mb{x}_{\minus i}}_k^*}^t}^2}{T^2}}) \leq 2\exp(- \frac{T \epsilon^2}{ \budget^2 \sigma^2 }) 
		\end{align*}
		$R_3$ is a sub-exponential random variable (check Lemma~\ref{lem: xew support} and \ref{lem:subexponential}). Thus, for $0 < \epsilon < \sigma^2$.
		\begin{align*}
		\prob(|R_3| \geq \epsilon) \leq 2\exp(- \frac{T \epsilon^2}{ 8 \sigma^4 }) 
		\end{align*}

	Now,
	\begin{align*}
	\| \res{\bHhat}{\supp^c \supp} - \res{\bH}{\supp^c \supp} \|_{\BIO} = \max_{i \in \supp^c} (\sum_{j \in \ind(i)} \sum_{k= 1}^k |Z_{jk}|)
	\end{align*}
	Combining the results for random variables $R_1, R_2$ and $R_3$ and applying union bound, we get
	\begin{align*}
	&\prob[|Z_{jk}| \geq \epsilon] \leq 2 D \exp(\frac{-\epsilon^2 T}{8 f(\sigma, \budget)}) 
	\end{align*}
	where $D > 0$ is a constant and $f(\sigma, \budget) = \max(\frac{\budget^2\sigma^2}{8}, \sigma^4)$.
	Taking $\epsilon = \frac{\delta}{k^2|\supp|} $ for any $i \in \supp^c$.
	\begin{align*}
	&\prob(|Z_{jk}| \geq \frac{\delta}{k^2|\supp|} ) \leq 2 D \exp(\frac{-\delta^2 T}{8 f(\sigma, \budget) k^4 |\supp|^2})
	\end{align*}
	Using the union bound over $\forall i \in \supp^c, j \in \ind(i), \forall l \in \supp, k \in \ind(l)$ we can write,
	\begin{align*}
	\begin{split}
	&\prob(\| \res{\bHhat}{\supp^c \supp} - \res{\bH}{\supp^c \supp} \|_{\BIO} \geq \delta ) \leq  |\supp^c| |\supp| k^2  2 D \exp(\frac{-\delta^2 T}{8  f(\sigma, \budget) k^4 |\supp|^2}) \\
	&\leq 2D \exp(\frac{-\delta^2 T}{8  f(\sigma, \budget) k^4 |\supp|^2} + \log k^2 |\supp^c| |\supp| )
	\end{split}
	\end{align*}
	Similarly we can prove equation~\eqref{eq:Teq2},
	\begin{align*}
	&\prob(\| \res{\bHhat}{\supp \supp} - \res{\bH}{\supp \supp} \|_{\infty, \infty} \geq \delta ) \leq k^2|\supp|^2 \prob(|Z_{jk}| \geq \frac{\delta}{k|\supp|}  ) \leq  2D \exp(\frac{-\delta^2 T}{8 f(\sigma, \budget) k^2 |\supp|^2  } + 2 \log k |\supp|)
	\end{align*}
	Now we prove equation~\eqref{eq:Teq3}. Note that,
	\begin{align*}
	\begin{split}
	&\|(\res{\bHhat}{\supp \supp})^{-1} - (\res{\bH}{\supp \supp})^{-1}\|_{\infty, \infty} = \\
	&\|(\res{\bH}{\supp \supp})^{-1} [\res{\bH}{\supp \supp}
	- \res{\bHhat}{\supp \supp}] (\res{\bHhat}{\supp \supp})^{-1}\|_{\infty, \infty} \\
	&\leq \sqrt{k|\supp|} \|(\res{\bH}{\supp \supp})^{-1} [\res{\bH}{\supp \supp} - \res{\bHhat}{\supp \supp}] (\res{\bHhat}{\supp \supp})^{-1}\|_{2,2} \\
	&\leq \sqrt{k|\supp|} \|(\res{\bH}{\supp \supp})^{-1}\|_{2,2} \|[\res{\bH}{\supp \supp} - \res{\bHhat}{\supp \supp}]\|_{2,2} \|(\res{\bHhat}{\supp \supp})^{-1}\|_{2,2} \\
	&\leq \frac{\sqrt{k|\supp|}}{C_{\min}}\|[\res{\bH}{\supp \supp} - \res{\bHhat}{\supp \supp}]\|_{2,2} \|(\res{\bHhat}{\supp \supp})^{-1}\|_{2,2}
	\end{split}
	\end{align*}
	Note that,
	$\prob(\Lambda_{\min}(\res{\bHhat}{\supp \supp}) \geq C_{\min} -
	\delta] \geq 1 - 2 \exp(- \frac{\delta^2 T }{8 k^2 |\supp|^2} + 2 \log k |\supp|) $. Taking $\delta = \frac{C_{\min}}{2}$, we get
	$\prob(\Lambda_{\min}(\res{\bHhat}{\supp \supp}) \geq \frac{C_{\min}}{2} ) \geq 1 - 2 \exp(- \frac{C_{\min}^2 T }{32 k^2|\supp|^2} + 2 \log k |\supp|) $.
	This means that,
	\begin{align}
	\label{eq:inverse_eigval}
	\begin{split}
	\prob(\|(\res{\bHhat}{\supp \supp})^{-1}\|_{2,2} \leq \frac{2}{C_{\min}} ) &\geq 1 - 2 \exp(- \frac{C_{\min}^2 T }{32 k^2 |\supp|^2} + \\
	&2 \log k |\supp|) \ .
	\end{split}
	\end{align}
	Furthermore,
	\begin{align*}
	\begin{split}
	\prob(\| \res{\bH}{\supp \supp} - \res{\bHhat}{\supp \supp} \|_{2,2} &\geq \epsilon) \leq 2 D \exp(- \frac{\epsilon^2 T }{8  f(\sigma, \budget)k^2 |\supp|^2} + 2 \log k |\supp|) \ 
	\end{split}
	\end{align*}
	Taking $\epsilon = \delta \frac{C_{\min}^2}{2\sqrt{k|\supp|}}$, we get:
	\begin{align*}
	\begin{split}
	\prob(\| \bH_{\supp \supp} - \bHhat_{\supp \supp} \|_{2,2} &\geq \delta \frac{C_{\min}^2}{2\sqrt{k|\supp|}}) \leq 2D \exp(- \frac{\delta^2 C_{\min}^4 T}{32k^3  f(\sigma, \budget)|\supp|^3} + 2 \log k |\supp|) 
	\end{split}
	\end{align*}
	\noindent It follows that,
	\begin{align*}
	\begin{split}
	\prob(	\|(\res{\bHhat}{\supp \supp})^{-1} - (\res{\bH}{\supp \supp})^{-1}\|_{\infty, \infty} \leq \delta) \geq& 1 -  
	2 \exp(- \frac{\delta^2 C_{\min}^4 N }{32 k^3 |\supp|^3} + 2 \log k |\supp| \\
	&- 2D \exp(- \frac{C_{\min}^2 T }{32  f(\sigma, \budget)k^2 |\supp|^2} + 2 \log k |\supp|)
	\end{split}
	\end{align*}
\end{proof} 
 
Now, we provide the detailed proof of Lemma \ref{lemma:mutual_incoherence}.
\paragraph{Lemma \ref{lemma:mutual_incoherence}}[Mutual incoherence in sample]
	\emph{If $\|\res{\bH}{\supp^c\supp} \res{\bH}{\supp \supp}^{-1}\|_{\BIO} \leq 1 - \alpha$ for $\alpha \in (0, 1]$ then $\|\res{\bHhat}{\supp^c \supp}\res{\bHhat}{\supp \supp}^{-1}\|_{\BIO} \leq 1 - \frac{\alpha}{2}$} with probability at least $ 1 - \mathcal{O}( \exp(\frac{-K T}{ k^5  |\supp|^3} + \log k |\supp^c| + \log k |\supp| )) $ for some $K > 0$. 
\begin{proof}
	We can rewrite $\bHhat_{\supp^c \supp}(\bHhat_{\supp \supp})^{-1}$ as the sum of four terms defined as:
	\begin{align}
	\label{eq:sample mutual incoherence}
	\begin{split}
	&\res{\bHhat}{\supp^c \supp}(\res{\bHhat}{\supp \supp})^{-1} = T_1 + T_2 + T_3 + T4  \\
	&\|\res{\bHhat}{\supp^c \supp}(\res{\bHhat}{\supp \supp})^{-1}\|_{\BIO} \leq \|T_1\|_{\BIO} + \|T_2\|_{\BIO} +  \|T_3\|_{\BIO} + \|T_4\|_{\BIO}
	\end{split}
	\end{align}
	where,
	\begin{align*}
	T_1 &\triangleq \res{\bH}{\supp^c \supp}[(\res{\bHhat}{\supp \supp})^{-1} - \res{\bH}{\supp \supp}^{-1} ] \\
	T_2 &\triangleq [\res{\bHhat}{\supp^c \supp} - \res{\bH}{\supp^c \supp}]\res{\bH}{\supp \supp}^{-1}\\
	T_3 &\triangleq [\res{\bHhat}{\supp^c \supp} - \res{\bH}{\supp^c \supp}][(\res{\bHhat}{\supp \supp})^{-1} - \res{\bH}{\supp \supp}^{-1} ] \\
	T_4 &\triangleq \res{\bH}{\supp^c \supp}\res{\bH}{\supp \supp}^{-1}
	\end{align*}
	and each $T_i$ is treated as a row-partitioned block matrix of $ |\supp^c|$ blocks with each block containing $k$ rows.
	From Mutual incoherence Assumption, it is clear that $\|T_4\|_{\BIO} \leq 1 - \alpha$. We control the other three terms by using results from Lemma~\ref{lem:control_T moved}.

	{\bf Controlling the first term of equation~\eqref{eq:sample mutual incoherence}.} 
	We can write $T_1$ as,
	\begin{align*}
	&T_1 = - \res{\bH}{\supp^c \supp} (\res{\bH}{\supp \supp})^{-1}[\res{\bHhat}{\supp \supp} - \res{\bH}{\supp \supp}](\res{\bHhat}{\supp \supp})^{-1} 
	\end{align*}
	\noindent then,
	\begin{align*}
	\begin{split}
	\|T_1\|_{\BIO} &= \| \res{\bH}{\supp^c \supp} (\res{\bH}{\supp \supp})^{-1}[\res{\bHhat}{\supp \supp} - \res{\bH}{\supp \supp}](\res{\bHhat}{\supp \supp})^{-1} \|_{\BIO} \\
	&\leq \| \res{\bH}{\supp^c \supp} (\res{\bH}{\supp \supp})^{-1} \|_{\BIO} \|[\res{\bHhat}{\supp \supp} - \res{\bH}{\supp \supp}]\|_{\infty, \infty} \|(\res{\bHhat}{\supp \supp})^{-1} \|_{\infty, \infty} \\
	&\leq (1 - \alpha) \|[\res{\bHhat}{\supp \supp} - \res{\bH}{\supp \supp}]\|_{\infty, \infty} \sqrt{k|\supp|}\|(\res{\bHhat}{\supp \supp})^{-1} \|_{2, 2}
	\end{split}
	\end{align*}
	Now using equation~\eqref{eq:inverse_eigval} and equation~\eqref{eq:Teq2} with $\delta = \frac{\alpha C_{\min}}{12\sqrt{k|\supp|}(1 - \alpha)}$ we can say that,
	\begin{align*}
	\begin{split}
	\prob[\|T_1\|_{\BIO} \leq \frac{\alpha}{6} ] \geq&  1 -  2 \exp(- \frac{C_{\min}^2 T }{32 k^2 |\supp|^2} + 2 \log k |\supp|)  - 2 \exp(-K\frac{T\alpha^2C_{\min}^2}{144 (1-\alpha)^2 k^3 |\supp|^3} + \\
	& 2 \log k |\supp|)  
	\end{split}
	\end{align*}
	{\bf Controlling the second term of equation~\eqref{eq:sample mutual incoherence}.}
	We can write $\|T_2\|_{\BIO}$ as,
	\begin{align*}
	\begin{split}
	\|T_2\|_{\BIO} &= \|  [\res{\bHhat}{\supp^c \supp} - \res{\bH}{\supp^c \supp}]\res{\bH}{\supp \supp}^{-1} \|_{\BIO} \\
	&\leq \|  [\res{\bHhat}{\supp^c \supp} - \res{\bH}{\supp^c \supp}] \|_{\BIO} \|\res{\bH}{\supp \supp}^{-1} \|_{\infty, \infty}\\
	&\leq \|  [\res{\bHhat}{\supp^c \supp} - \res{\bH}{\supp^c \supp}] \|_{\BIO} \sqrt{k|\supp|}\|\res{\bH}{\supp \supp}^{-1} \|_{2,2}\\
	&\leq \frac{\sqrt{k|\supp|}}{C_{\min}}\|  [\res{\bHhat}{\supp^c \supp} - \res{\bH}{\supp^c \supp}] \|_{\BIO} 
	\end{split}
	\end{align*}
	Using equation~\eqref{eq:Teq1} with $\delta = \frac{\alpha C_{\min}}{6 \sqrt{k|\supp|}}$ we get,
	\begin{align*}
	\begin{split} 
	&\prob(\|T_2\|_{\BIO} \leq \frac{\alpha}{6} ) \geq  1 - 2 D \exp(\frac{-\alpha^2 C_{\min}^2 T}{288  f(\sigma, \budget) k^5 |\supp|^3} +  \log k |\supp^c| + \log k |\supp| )
	\end{split}
	\end{align*}
	{\bf Controlling the third term of equation~\eqref{eq:sample mutual incoherence}.} 
	We can write $\|T_3\|_{\BIO}$ as, 
	\begin{align*}
	\begin{split}
	&\|T_3\|_{\BIO} \leq 
	\|[\res{\bHhat}{\supp^c \supp} - \res{\bH}{\supp^c \supp}] \|_{\BIO} \|[(\res{\bHhat}{\supp \supp})^{-1} - \res{\bH}{\supp \supp}^{-1} ]\|_{\infty, \infty} 
	\end{split}
	\end{align*}
	\noindent Using equation~\eqref{eq:Teq1} and~\eqref{eq:Teq3} both with $\delta = \sqrt{\frac{\alpha}{6}}$, we get
	\begin{align*}
	\begin{split}
	&\prob(\|T_3\|_{\BIO} \leq \frac{\alpha}{6} )\geq  1 -  2 D\exp(- \frac{\delta^2 C_{\min}^4 T }{32  f(\sigma, \budget) k^3 |\supp|^3} + 2 \log k |\supp|) - 2D \exp(- \frac{C_{\min}^2 T }{32 f(\sigma, \budget) k^2 |\supp|^2} + 2 \log k |\supp|) \\
	&- 2 D \exp(\frac{-\alpha T}{48 f(\sigma, \budget) (k^3 |\supp|)^2} + \log k |\supp^c| + \log k |\supp| )
	\end{split}
	\end{align*}
	\noindent Putting everything together we get,
	\begin{align*}
	\begin{split}
	&\prob[\| \res{\bHhat}{\supp^c \supp} \res{\bHhat}{\supp \supp}^{-1} \|_{\BIO} \leq 1 - \frac{\alpha}{2}] \geq 1 - \mathcal{O}( \exp(\frac{-K T}{ f(\sigma, \budget) k^5 |\supp|^3} + \log k |\supp^c| + \log k |\supp| )
	\end{split}
	\end{align*}
	which approaches $1$ as long as we have $N > \mathcal{O}(k^5  d^3 \log nk)$
\end{proof}

\subsection{Proof of Lemma \ref{lem: xew support}}
\paragraph{Lemma \ref{lem: xew support}}[Bound on \emph{$\| \res{\hE{\mb{x}_{\minus i} {\mb{e}_{\minus i}}^\T}}{\supp \supp} \res{\mb{W}_{i.}^*}{\supp\cdot} \|_{\infty, 2}$}]
\emph{For some $\epsilon_1 > 0, 0 < \epsilon_2 < 8$ and $\epsilon_3 < 8 \frac{\sqrt{|\supp|} \max_{ij} |\mb{W}_{ij}^*|}{\min_{ij}|\mb{W}_{ij}^*|}$,
		\begin{align}
		\begin{split}
		&\prob\big(\| \res{\hE{\mb{x}_{\minus i} {\mb{e}_{\minus i}}^\T   }}{\supp \supp} \res{\mb{W}_{i.}^*}{\supp\cdot} \|_{\infty, 2} >  \epsilon_1 + \sigma^2 (1 + \epsilon_2 +\epsilon_3 ) \max_{i \in \ind(l), l \in \supp} \sqrt{\sum_{j=1}^{k} |\mb{W}_{ij}^*|^2} \big) \leq \\
		&\exp(-\frac{\epsilon_1^2 T}{2k \sigma^2 \budget^2 \max_{j}\sum_{k=1}^{|\supp|} {\mb{W}_{kj}^*}^2 } + \log(2 k^2 |\supp|) ) + k^2 |\supp| \exp(-\frac{T\epsilon_2^2}{64}) + \sum_{\substack{i \in \ind(l)\\ l \in \supp}} \sum_{j=1}^{k} \exp(-\frac{T\epsilon_3^2 |\mb{W}_{ij}^*|}{64 |\sqrt{\sum_{\substack{k=1 \\ k \ne i}}^{|\supp|}  {\mb{W}_{kj}^*}^2}|})
		\end{split}
		\end{align}}
\begin{proof}
	Note that,
	\begin{align}
	\begin{split}
	&\| \res{\hE{\mb{x}_{\minus i} {\mb{e}_{\minus i}}^\T}}{\supp \supp} \res{\mb{W}_{i.}^*}{\supp\cdot} \|_{\infty, 2} \quad\quad\phantom{aaaaaaaaaaaaaaaaaaaaaa} \\
	&= \| \res{\hE{\mb{x}_{\minus i}^* {\mb{e}_{\minus i}}^\T  + \mb{e}_{\minus i} {\mb{e}_{\minus i}}^\T  }}{\supp \supp} \res{\mb{W}_{i.}^*}{\supp\cdot} \|_{\infty, 2} \\
	&\leq \| \res{\hE{\mb{x}_{\minus i}^* {\mb{e}_{\minus i}}^\T   }}{\supp \supp} \res{\mb{W}_{i.}^*}{\supp\cdot} \|_{\infty, 2} +  \| \res{\hE{\mb{e}_{\minus i} {\mb{e}_{\minus i}}^\T  }}{\supp \supp} \res{\mb{W}_{i.}^*}{\supp\cdot} \|_{\infty, 2}  
	\end{split}
	\end{align}
	Again, we will bound both terms separately.
	\paragraph{Bound on $\| \res{\hE{\mb{x}_{\minus i}^* {\mb{e}_{\minus i}}^\T   }}{\supp \supp} \res{\mb{W}_{i.}^*}{\supp\cdot} \|_{\infty, 2}$.} 
	For simplicity, let $\res{\mb{x}_{\minus i}^*}{\supp} = \mb{y}^* \in \real^{|\supp|k \times 1}$,  $\res{\mb{e}_{\minus i}}{\supp} = \mb{u} \in \real^{|\supp|k \times 1}$ and $\res{\mb{W}_{i.}^*}{\supp\cdot} = \mb{W}^* \in \real^{|\supp|k\times k}$. We define a random variable $R$ and then,
	
	\begin{align}
	\begin{split}
	R &\triangleq \hE{\mb{y}^* \mb{u}^\T \mb{W}^*}_{ij} \\
	&= \hE{\sum_{k=1}^{|\supp|} \mb{y}_i^* \mb{u}_k \mb{W}_{kj}^*} \\
	&= \frac{1}{T} \sum_{t=1}^T \big( \sum_{k=1}^{|\supp|} {\mb{y}_i^*}^t \mb{u}_k^t \mb{W}_{kj}^*  \big)
	\end{split}
	\end{align}
	
	For a given ${\mb{y}_i^t}^* , \forall t\in \seq{T}$, $R$ is a sub-Gaussian random variable with $0$ mean and parameter $\frac{\sigma^2}{T^2} \sum_{t=1}^T \sum_{k=1}^{|\supp|} \big( {\mb{y}_i^*}^t \mb{W}_{kj}^* \big)^2$. Thus, for some $\epsilon_1 > 0$, we can use a tail bound on a sub-Gaussian random variable:
	
	\begin{align}
	\begin{split}
	\prob_{\cdot\mid {\mb{y}_i^t}^* }(|R| > \epsilon_1) &\leq 2 \exp(-\frac{\epsilon_1^2 T^2}{2\sigma^2 \sum_{t=1}^T \sum_{k=1}^{|\supp|} \big( {\mb{y}_i^*}^t \mb{W}_{kj}^* \big)^2 }) \\
	&\leq 2 \exp(-\frac{\epsilon_1^2 T}{2\sigma^2 \budget^2 \sum_{k=1}^{|\supp|} {\mb{W}_{kj}^*}^2 })
	\end{split}
	\end{align}
	where last inequality holds because ${\mb{y}_i^*}^t \leq \budget, \forall t \in \seq{T}$. Therefore,
	\begin{align}
	\begin{split}
	\prob(|R| > \epsilon_1) &= \texttt{E}_{{\mb{y}_i^t}^* }\big( \prob_{\cdot\mid {\mb{y}_i^t}^* }(|R| > \epsilon_1) \big) \\
	&\leq 2 \exp(-\frac{\epsilon_1^2 T}{2\sigma^2 \budget^2 \sum_{k=1}^{|\supp|} {\mb{W}_{kj}^*}^2 })
	\end{split}
	\end{align}
	Taking union bound across $i \in \ind(l), \forall l\in \supp$ and $j \in \seq{k}$, we get
	\begin{align}
	\begin{split}
	&\prob(\| \res{\hE{\mb{x}_{\minus i}^* {\mb{e}_{\minus i}}^\T   }}{\supp \supp} \res{\mb{W}_{i.}^*}{\supp\cdot} \|_{\infty, 2} > \epsilon_1) \phantom{aaaaaaaaaaaaa}\\
	&\leq 2 k^2 |\supp| \exp(-\frac{\epsilon_1^2 T}{2k \sigma^2 \budget^2 \max_{j}\sum_{k=1}^{|\supp|} {\mb{W}_{kj}^*}^2 }) \\
	&= \exp(-\frac{\epsilon_1^2 T}{2k \sigma^2 \budget^2 \max_{j}\sum_{k=1}^{|\supp|} {\mb{W}_{kj}^*}^2 } + \log(2 k^2 |\supp|) )
	\end{split}
	\end{align}
	
	\paragraph{Bound on $\| \res{\hE{\mb{e}_{\minus i} {\mb{e}_{\minus i}}^\T   }}{\supp \supp} \res{\mb{W}_{i.}^*}{\supp\cdot} \|_{\infty, 2}$.}
	Again for simplicity, let  $\res{\mb{e}_{\minus i}}{\supp} = \mb{u} \in \real^{|\supp|k \times 1}$ and $\res{\mb{W}_{i.}^*}{\supp\cdot} = \mb{W}^* \in \real^{|\supp|k\times k}$.We define a random variable $R$ and then,
	
	\begin{align}
	\begin{split}
	R &\triangleq \hE{\mb{u}\mb{u}^\T \mb{W}^*}_{ij} \\
	&= \frac{1}{T} \sum_{t=1}^T \big( \sum_{k=1}^{|\supp|} \mb{u}_i^t \mb{u}_k^t \mb{W}^*_{kj} \big) \\
	&= \frac{1}{T} \sum_{t=1}^T \mb{u}_i^t \mb{u}_i^t \mb{W}_{ij}^* + \frac{1}{T} \sum_{t=1}^T \big( \sum_{\substack{k=1 \\ k \ne i}}^{|\supp|} \mb{u}_i^t \mb{u}_k^t \mb{W}_{kj}^* \big) \\
	\end{split}
	\end{align}
	
	We will bound the random variable $R$ in two steps. First, we will bound the random variable $R_1 \triangleq \frac{1}{T} \sum_{t=1}^T \mb{u}_i^t \mb{u}_i^t \mb{W}_{ij}^*$ and then we will bound the random variable $R_2 \triangleq \frac{1}{T} \sum_{t=1}^T \big( \sum_{\substack{k=1 \\ k \ne i}}^{|\supp|} \mb{u}_i^t \mb{u}_k^t \mb{W}_{kj}^* \big)$.
	\paragraph{Bound on $R1$}
	We observe that 
	\begin{align}
	\begin{split}
	R_1 = \frac{\sigma^2 \mb{W}_{ij}^*}{T} \sum_{t=1}^T \frac{\mb{u}_i^t}{\sigma} \frac{\mb{u}_i^t}{\sigma} 
	\end{split}
	\end{align}
	 We define a random variable $ y \triangleq \frac{\mb{u}_i^t}{\sigma} $. Note that $y$ is a sub-Gaussian random variable with $0$ mean and parameter $1$. We prove in Lemma \ref{lem:y^2 is subexponential} that $y^2$ is a sub-exponential random variable with parameter $(4\sqrt{2}, 4)$.
	Therefore, we can use a Bernstein type bound on $(\frac{\mb{u}_i^t}{\sigma})^2$, thus for some $\epsilon_2 \in (0, 8)$,
	\begin{align}
		\prob(\mid \frac{1}{T} \sum_{t=1}^T (\frac{\mb{u}_i^t}{\sigma} \frac{\mb{u}_i^t}{\sigma} - \E{\frac{\mb{u}_i^t}{\sigma} \frac{\mb{u}_i^t}{\sigma}} )\mid > \epsilon_2) \leq 2 \exp(-\frac{T\epsilon_2^2}{64})
	\end{align}
	Furthermore, note that $\E{\frac{\mb{u}_i^t}{\sigma} \frac{\mb{u}_i^t}{\sigma}} \leq 1$, thus
	\begin{align}
	\prob(\mid \frac{1}{T} \sum_{t=1}^T (\frac{\mb{u}_i^t}{\sigma} \frac{\mb{u}_i^t}{\sigma} - 1 )\mid > \epsilon_2) \leq 2 \exp(-\frac{T\epsilon_2^2}{64})
	\end{align}

	\begin{align}
	\begin{split}
	&\prob(\sigma^2 |\mb{W}_{ij}^*| \mid \frac{1}{T} \sum_{t=1}^T \frac{\mb{u}_i^t}{\sigma} \frac{\mb{u}_i^t}{\sigma} - 1 \mid > \sigma^2 |\mb{W}_{ij}^*| \epsilon_2) \leq 2 \exp(-\frac{T\epsilon_2^2}{64}) 
	\end{split}
	\end{align} 
	In other words,
	\begin{align}
	\begin{split}
	&\prob(R_1  < \sigma^2 |\mb{W}_{ij}^*| - \sigma^2 |\mb{W}_{ij}^*| \epsilon_2 \vee R_1  > \sigma^2 |\mb{W}_{ij}^*| +  \sigma^2 |\mb{W}_{ij}^*| \epsilon_2 ) \leq 2 \exp(-\frac{T\epsilon_2^2}{64}) 
	\end{split}
	\end{align} 
	\paragraph{Bound on $R2$}
	We observe that
	
	\begin{align}
	\begin{split}
	R_2 =  \sigma^2 \sqrt{\sum_{\substack{k=1 \\ k \ne i}}^{|\supp|}  {\mb{W}_{kj}^*}^2} \frac{1}{T} \sum_{t=1}^T \frac{\mb{u}_i^t}{\sigma} \frac{ \sum_{\substack{k=1 \\ k \ne i}}^{|\supp|}  \mb{u}_k^t \mb{W}_{kj}^* }{\sigma\sqrt{\sum_{\substack{k=1 \\ k \ne i}}^{|\supp|} {\mb{W}_{kj}^*}^2 }}
	\end{split}
	\end{align}	
	Here $\frac{\mb{u}_i^t}{\sigma}$ and $\frac{ \sum_{\substack{k=1 \\ k \ne i}}^{|\supp|}  \mb{u}_k^t \mb{W}_{kj}^* }{\sigma\sqrt{\sum_{\substack{k=1 \\ k \ne i}}^{|\supp|} {\mb{W}_{kj}^*}^2 }}$ are independent sub-Gaussian random variables with $0$ mean and parameter $1$. Thus, similar to Lemma \ref{lem:bound sup xe}, we can use a Bernstein type tail bound for the sum of sub-exponential random variables. For some $\epsilon_3 < 8$,
	
	\begin{align}
	\begin{split}
	&\prob(|\frac{1}{T} \sum_{t=1}^T \frac{\mb{u}_i^t}{\sigma} \frac{ \sum_{\substack{k=1 \\ k \ne i}}^{|\supp|}  \mb{u}_k^t \mb{W}_{kj}^* }{\sigma\sqrt{\sum_{\substack{k=1 \\ k \ne i}}^{|\supp|} {\mb{W}_{kj}^*}^2 }}| > \epsilon_3) \leq 2 \exp(-\frac{T\epsilon_3^2}{64}) 
	\end{split}
	\end{align}
	Or for some  $\epsilon_3 < 8\sigma^2 |\sqrt{\sum_{\substack{k=1 \\ k \ne i}}^{|\supp|}  {\mb{W}_{kj}^*}^2}| $,
	\begin{align}
	\begin{split}
	&\prob(|R_2| > \epsilon_3) \leq 2 \exp(-\frac{T\epsilon_3^2}{64\sigma^4 \sum_{\substack{k=1 \\ k \ne i}}^{|\supp|}  {\mb{W}_{kj}^*}^2})
	\end{split}
	\end{align}
	
	Combining the bounds on $R_1$ and $R_2$ and taking a union bound, we get
	
	\begin{align}
	\begin{split}
	&\prob(R < \sigma^2 |\mb{W}_{ij}^*| - \sigma^2 |\mb{W}_{ij}^*| \epsilon_2 - \epsilon_3 \vee R > \sigma^2 |\mb{W}_{ij}^*| + \sigma^2 |\mb{W}_{ij}^*| \epsilon_2 + \epsilon_3  ) \leq 2 \exp(-\frac{T\epsilon_2^2}{8})  \\
	&+ 2 \exp(-\frac{T\epsilon_3^2}{64\sigma^4 |\sum_{\substack{k=1 \\ k \ne i}}^{|\supp|}  {\mb{W}_{kj}^*}^2|})
	\end{split}
	\end{align}
	Or for some $\epsilon_3 < \frac{8 |\sqrt{\sum_{\substack{k=1 \\ k \ne i}}^{|\supp|}  {\mb{W}_{kj}^*}^2}|}{|\mb{W}_{ij}^*|}$,
	\begin{align}
	\begin{split}
	&\prob(R < \sigma^2 |\mb{W}_{ij}^*|(1 - \epsilon_2 - \epsilon_3) \vee R > \sigma^2 |\mb{W}_{ij}^*| (1 + \epsilon_2 + \epsilon_3 ) \leq 2 \exp(-\frac{T\epsilon_2^2}{8})  \\
	&+ 2 \exp(-\frac{T\epsilon_3^2 |\mb{W}_{ij}^*|}{64 |\sum_{\substack{k=1 \\ k \ne i}}^{|\supp|}  {\mb{W}_{kj}^*}^2|})
	\end{split}
	\end{align}
	Taking one sided union bound across $i \in \ind(l), \forall l\in \supp$ and $j \in \seq{k}$, we get
	
	\begin{align}
	\begin{split}
	&\prob(\| \res{\hE{\mb{e}_{\minus i} {\mb{e}_{\minus i}}^\T   }}{\supp \supp} \res{\mb{W}_{i.}^*}{\supp\cdot} \|_{\infty, 2} > \sigma^2 (1 + \epsilon_2 + \epsilon_3 ) \max_{i \in \ind(l), l \in \supp} \sqrt{\sum_{j=1}^{k} |\mb{W}_{ij}^*|^2} ) \\
	&\leq k^2 |\supp| \exp(-\frac{T\epsilon_2^2}{8}) + \sum_{i \in \ind(l), l \in \supp} \sum_{j=1}^{k} \exp(-\frac{T\epsilon_3^2 |\mb{W}_{ij}^*|}{64 |\sum_{\substack{k=1 \\ k \ne i}}^{|\supp|}  {\mb{W}_{kj}^*}^2|})
	\end{split}
	\end{align}
	
\end{proof}

\subsection{Proof of Lemma \ref{lem:bound sup xe}}
\paragraph{Lemma \ref{lem:bound sup xe}}[Bound on \emph{$\|   \hE{\res{\mb{x}_{\minus i}}{\supp\cdot} {\mb{e}_i}^\T} \|_{\infty, 2} $}]
\emph{	For some $\epsilon_4 > 0$ and $\epsilon_5 < 8 \sqrt{k} \sigma^2$,
	\begin{align}
	\begin{split}
	\prob(\| \hE{\res{\mb{x}_{\minus i}}{\supp\cdot} {\mb{e}_i}^\T}\|_{\infty, 2} \geq \epsilon_4 + \epsilon_5) \leq 
	\exp(- \frac{\epsilon_4^2 T}{2 k \sigma^2 \budget^2 } + \log(2k^2|\supp|)) + \exp(- \frac{\epsilon_5^2 T}{64 \sqrt{k} \sigma^2 } + \log(2k^2|\supp|))
	\end{split}
	\end{align}}
\begin{proof}
	Note that $\res{\mb{x}_{\minus i}}{\supp\cdot} = \res{\mb{x}_{\minus i}^*}{\supp\cdot} + \res{\mb{e}_{\minus i}}{\supp\cdot} $. Thus,
	
	\begin{align}
	\begin{split}
	&\|   \hE{\res{\mb{x}_{\minus i}}{\supp\cdot} {\mb{e}_i}^\T} \|_{\infty, 2} = \|   \hE{(\res{\mb{x}_{\minus i}^*}{\supp\cdot} + \res{\mb{e}_{\minus i}}{\supp\cdot}) {\mb{e}_i}^\T} \|_{\infty, 2}\\
	&\leq \|   \hE{(\res{\mb{x}_{\minus i}^*}{\supp\cdot} {\mb{e}_i}^\T}\|_{\infty, 2} + \| \hE{\res{\mb{e}_{\minus i}}{\supp\cdot}) {\mb{e}_i}^\T} \|_{\infty, 2}
	\end{split}
	\end{align}
	
	We will bound both the terms separately. 
	\paragraph{Bound on $\| \hE{\res{\mb{x}_{\minus i}^*}{\supp\cdot} {\mb{e}_i}^\T}\|_{\infty, 2}$.}
	For simplicity, let $\res{\mb{x}_{\minus i}^*}{\supp\cdot} = \mb{y}^* \in \real^{|\supp|k \times 1}$ and $\mb{e}_i = \mb{u} \in \real^k$ and we define a random variable $R$
	
	\begin{align}
	\begin{split}
	R &\triangleq \hE{\mb{y}^* \mb{u}^\T}_{ij} \\
	&= \hE{\mb{y}^*_{i1} \mb{u}_j} \\
	&= \frac{1}{T} \sum_{t=1}^T ({\mb{y}_{i1}^*}^t \mb{u}_j^t) \\
	&= \sum_{t=1}^T R^t
	\end{split}
	\end{align}
	Then for a given $\mb{y}_{i1}^t$, random variable $R^t$ is a sub-Gaussian random variable with $0$ mean and parameter $\frac{{\mb{y}_{i1}^t}^2 \sigma^2}{T^2}$. Correspondingly $R$ is a sub-Gaussian random variable with $0$ mean and parameter $\sigma^2 \frac{\sum_{t=1}^T {\mb{y}_{i1}^t}^2}{T^2}$. Using the tail bound for the sub-Gaussian variable for some $\epsilon_4 > 0$, we can write
	\begin{align}
	\begin{split}
	\prob_{.\mid \mb{y}_{i1}^t}(\mid R \mid > \epsilon_4) &\leq 2 \exp( - \frac{\epsilon_4^2}{2 \sigma^2 \frac{\sum_{t=1}^T {\mb{y}_{i1}^t}^2}{T^2} }) \\
	&\leq 2 \exp(- \frac{\epsilon_4^2 T}{2 \sigma^2 \budget^2 })
	\end{split}
	\end{align}   
	where last inequality follows by noting that ${\mb{y}_{i1}^t}^2 \leq \budget^2$. Thus,
	\begin{align}
	\begin{split}
	\prob(\mid \hE{\mb{y}^* \mb{u}^\T}_{ij} \mid > \epsilon_4) &= \texttt{E}_{\mb{y}_{i1}^t}\big( \prob_{.\mid \mb{y}_{i1}^t}(\mid \hE{\mb{y}^* \mb{u}^\T}_{ij} \mid > \epsilon_4) \big) \\
	&\leq 2 \exp(- \frac{\epsilon_4^2 T}{2 \sigma^2 \budget^2 })
	\end{split}
	\end{align}
	
	Now,
	\begin{align}
	\begin{split}
	\| \hE{\mb{y}^* \mb{u}^\T}\|_{\infty, 2} = \max_{l \in \supp} \max_{i \in \ind(l)} \| \hE{\mb{y}^*_{l1} \mb{u}^\T} \|_2
	\end{split}
	\end{align}
	Taking union bound across $i \in \ind(l), \forall l\in \supp$ and $j \in \seq{k}$, we get
	
	\begin{align}
	\begin{split}
	\prob(\| \hE{\res{\mb{x}_{\minus i}^*}{\supp\cdot} {\mb{e}_i}^\T}\|_{\infty, 2} &\geq \epsilon_4) \leq 2 k^2 \mid \supp \mid \exp(- \frac{\epsilon_4^2 T}{2 k \sigma^2 \budget^2 }) \\
	&= \exp(- \frac{\epsilon_4^2 T}{2 k \sigma^2 \budget^2 } + \log(2k^2|\supp|))
	\end{split}
	\end{align}
	
	\paragraph{Bound on $\| \hE{\res{\mb{e}_{\minus i}}{\supp\cdot} {\mb{e}_i}^\T}\|_{\infty, 2}$.}
	Again for simplicity, let $\res{\mb{e}_{\minus i}}{\supp\cdot} = \mb{v} \in \real^{|\supp|k \times 1}$ and $\mb{e}_i = \mb{u} \in \real^k$ and we define a random variable $R$
	\begin{align}
	\begin{split}
	R &\triangleq \hE{\mb{v} \mb{u}^\T}_{ij} \\
	&= \frac{1}{T} \sum_{t=1}^T \mb{v}_{i1}^t \mb{u}_j^t 
	\end{split}
	\end{align}
	Note that $\mb{v}_{i1}^t$ and $\mb{u}_j^t$ are independent sub-Gaussian random variables with $0$ mean and $\sigma^2$ parameter. We will use Lemma \ref{lem:subexponential} to get a tail bound on the random variable $R$. 

	Now, $\E{R} = 0$ and for some $\epsilon_5 > 0$,
	
	\begin{align}
	\begin{split}
	\prob(\mid \frac{1}{T} \sum_{t=1}^T \mb{v}_{i1}^t \mb{u}_j^t  \mid > \epsilon_5) = \prob(\sigma^ 2\mid \frac{1}{T} \sum_{t=1}^T \frac{\mb{v}_{i1}^t}{\sigma} \frac{\mb{u}_j^t}{\sigma}  \mid > \epsilon_5)
	\end{split}
	\end{align}
	Here $\frac{\mb{v}_{i1}^t}{\sigma}$ and $\frac{\mb{u}_j^t}{\sigma}$ are sub-Gaussian random variables with parameter $1$. Thus using result from Lemma \ref{lem:subexponential}, we can use a Bernstein tail bound for the sum of sub-exponential random variables and write,
	\begin{align}
	\begin{split}
	\prob(\mid \frac{1}{T} \sum_{t=1}^T \mb{v}_{i1}^t \mb{u}_j^t  \mid > \epsilon_5) \leq 2 \exp(- \frac{T \epsilon_5^2}{64 \sigma^4}), \forall \epsilon_5 < 8 \sigma^2  
	\end{split}
	\end{align}
	Again, taking union bound across $i \in \ind(l), \forall l\in \supp$ and $j \in \seq{k}$, for all $\epsilon_5 < 8\sqrt{k} \sigma^2$ we get
	\begin{align}
	\begin{split}
	\prob(\| \hE{\res{\mb{e}_{\minus i}}{\supp\cdot} {\mb{e}_i}^\T}\|_{\infty, 2} &\geq \epsilon_5) \leq 2 k^2 \mid \supp \mid \exp(- \frac{\epsilon_5^2 T}{64 \sqrt{k} \sigma^4 }) \\
	&= \exp(- \frac{\epsilon_5^2 T}{64 \sqrt{k} \sigma^4 } + \log(2k^2|\supp|))
	\end{split}
	\end{align}
\end{proof}

\subsection{Proof of Lemma \ref{lem:bound x-ie-i}}
\paragraph{Lemma \ref{lem:bound x-ie-i}}[Bound on \emph{$\| \hE{\res{\mb{x}_{\minus i} {\mb{e}_{\minus i}}^\T}{\supp^c \supp}} \res{\mb{W}_{i.}^*}{\supp\cdot} \|_{\BIF}$}]
	\emph{For some $\epsilon_6 > 0, 0 <\epsilon_7 < 8$ and $ \epsilon_8 <  8 \frac{\sqrt{|\supp^c|} \max_{ij} |\mb{W}_{ij}^*|}{\min_{ij}|\mb{W}_{ij}^*|}$,
		\begin{align}
		\begin{split}
		&\prob(\| \res{\hE{\mb{x}_{\minus i} {\mb{e}_{\minus i}}^\T   }}{\supp^c \supp} \res{\mb{W}_{i.}^*}{\supp\cdot} \|_{\BIF} > \epsilon_6 + \sigma^2 (1 + \epsilon_7 + \epsilon_8 ) \max_{l \in \supp^c} \sqrt{\sum_{i \in \ind(l)}\sum_{j=1}^{k} |\mb{W}_{ij}^*|^2} ) \\
		&\leq\exp(-\frac{\epsilon_6^2 T}{2k^2 \sigma^2 \budget^2 \max_{j}\sum_{k=1}^{|\supp^c|} {\mb{W}_{kj}^*}^2 } + \log(2 k^2 |\supp^c|) ) +
		k^2 |\supp^c| \exp(-\frac{T\epsilon_7^2}{64}) +\\
		& \sum_{i \in \ind(l), l \in \supp^c} \sum_{j=1}^{k} \exp(-\frac{T\epsilon_8^2 |\mb{W}_{ij}^*|}{64 |\sqrt{\sum_{\substack{k=1 \\ k \ne i}}^{|\supp^c|}  {\mb{W}_{kj}^*}^2}|})
		\end{split}
		\end{align}}
\begin{proof}
	Note that,
	\begin{align}
	\begin{split}
	&\| \hE{\res{\mb{x}_{\minus i}^t {\mb{e}_{\minus i}^t}^\T}{\supp^c \supp}} \res{\mb{W}_{i.}^*}{\supp\cdot} \|_{\BIF} \leq \| \res{\hE{\mb{x}_{\minus i}^* {\mb{e}_{\minus i}}^\T   }}{\supp^c \supp} \res{\mb{W}_{i.}^*}{\supp\cdot} \|_{\BIF} \\
	& +  \| \res{\hE{\mb{e}_{\minus i} {\mb{e}_{\minus i}}^\T  }}{\supp^c \supp} \res{\mb{W}_{i.}^*}{\supp\cdot} \|_{\BIF} 
	\end{split}
	\end{align}
	We can follow the exact same argument of Lemma \ref{lem: xew support} to bound the above two terms until we take the union bound. This time we will take union bound across $i \in \supp^c$ and $j \in \seq{k \times k}$, we get
	\begin{align}
	\begin{split}
	&\prob(\| \res{\hE{\mb{x}_{\minus i}^* {\mb{e}_{\minus i}}^\T   }}{\supp^c \supp} \res{\mb{W}_{i.}^*}{\supp\cdot} \|_{\BIF} > \epsilon_6) \\
	&\leq 2 k^2 |\supp^c| \exp(-\frac{\epsilon_6^2 T}{2k^2 \sigma^2 \budget^2 \max_{j}\sum_{k=1}^{|\supp^c|} {\mb{W}_{kj}^*}^2 }) \\
	&= \exp(-\frac{\epsilon_6^2 T}{2k^2 \sigma^2 \budget^2 \max_{j}\sum_{k=1}^{|\supp^c|} {\mb{W}_{kj}^*}^2 } + \log(2 k^2 |\supp^c|) )
	\end{split}
	\end{align}
	and
	\begin{align}
	\begin{split}
	&\prob(\| \res{\hE{\mb{e}_{\minus i} {\mb{e}_{\minus i}}^\T   }}{\supp^c \supp} \res{\mb{W}_{i.}^*}{\supp\cdot} \|_{\BIF} > \sigma^2 (1 + \epsilon_7 + \epsilon_8 ) \max_{l \in \supp^c} \sqrt{\sum_{i \in \ind(l)}\sum_{j=1}^{k} |\mb{W}_{ij}^*|^2} ) \\
	&\leq k^2 |\supp^c| \exp(-\frac{T\epsilon_7^2}{64}) +\sum_{i \in \ind(l), l \in \supp^c} \sum_{j=1}^{k} \exp(-\frac{T\epsilon_8^2 |\mb{W}_{ij}^*|}{64 |\sqrt{\sum_{\substack{k=1 \\ k \ne i}}^{|\supp^c|}  {\mb{W}_{kj}^*}^2}|})
	\end{split}
	\end{align}
	for some $\epsilon_6, \epsilon_7$ and $ \epsilon_8 <  2 \frac{\sqrt{|\supp^c|} \max_{ij} |\mb{W}_{ij}^*|}{\min_{ij}|\mb{W}_{ij}^*|}$.
\end{proof}

\subsection{Proof of Lemma \ref{lem:bound x-ie_i}}
\paragraph{Lemma \ref{lem:bound x-ie_i}}[Bound on \emph{$\| \hE{\res{\mb{x}_{\minus i}}{\supp^c\cdot} {\mb{e}_i}^\T} \|_{\BIF} $}]
\emph{For some $\epsilon_9 > 0$ and $\epsilon_{10} < 8 k \sigma^2$,
	\begin{align}
	\begin{split}
	&\prob(\| \hE{\res{\mb{x}_{\minus i}}{\supp^c\cdot} {\mb{e}_i}^\T}\|_{\BIF} \geq \epsilon_9 + \epsilon_{10}) \leq \exp(- \frac{\epsilon_9^2 T}{2 k^2 \sigma^2 \budget^2 } + \log(2k^2|\supp^c|)) + \exp(- \frac{\epsilon_{10}^2 T}{64 k \sigma^2 } + \log(2k^2|\supp^c|))
	\end{split}
	\end{align}}
\begin{proof}
	Again note that $\res{\mb{x}_{\minus i}}{\supp^c\cdot} = \res{\mb{x}_{\minus i}^*}{\supp^c\cdot} + \res{\mb{e}_{\minus i}}{\supp^c\cdot} $. Thus,
	
	\begin{align}
	\begin{split}
	&\|   \hE{\res{\mb{x}_{\minus i}}{\supp^c\cdot} {\mb{e}_i}^\T} \|_{\BIF} = \|   \hE{(\res{\mb{x}_{\minus i}^*}{\supp^c\cdot} + \res{\mb{e}_{\minus i}}{\supp^c\cdot}) {\mb{e}_i}^\T} \|_{\BIF}\\
	&\leq \|   \hE{(\res{\mb{x}_{\minus i}^*}{\supp^c\cdot} {\mb{e}_i}^\T}\|_{\BIF} + \| \hE{\res{\mb{e}_{\minus i}}{\supp^c\cdot}) {\mb{e}_i}^\T} \|_{\BIF}
	\end{split}
	\end{align}
	
	Like in Lemma \ref{lem:bound sup xe}, we can bound both the terms separately using similar arguments. The only change would be that this time we will take union bound across $i \in \supp^c$ and $j \in \seq{k \times k}$, we get
	
	\begin{align}
	\begin{split}
	\prob(\| \hE{\res{\mb{x}_{\minus i}^*}{\supp^c\cdot} {\mb{e}_i}^\T}\|_{\BIF} &\geq \epsilon_9) \leq 2 k^2 \mid \supp^c \mid \exp(- \frac{\epsilon_9^2 T}{2 k^2 \sigma^2 \budget^2 }) \\
	&= \exp(- \frac{\epsilon_9^2 T}{2 k^2 \sigma^2 \budget^2 } + \log(2k^2|\supp^c|))
	\end{split}
	\end{align}
	and similarly for $\epsilon_{10} < 2 k \sigma^2$,
	\begin{align}
	\begin{split}
	\prob(\| \hE{\res{\mb{e}_{\minus i}}{\supp^c\cdot} {\mb{e}_i}^\T}\|_{\BIF} &\geq \epsilon_{10}) \leq 2 k^2 \mid \supp^c \mid \exp(- \frac{\epsilon_{10}^2 T}{64 k \sigma^2 }) \\
	&= \exp(- \frac{\epsilon_{10}^2 T}{64 k \sigma^2 } + \log(2k^2|\supp^c|))
	\end{split}
	\end{align}
\end{proof}

\section{Proofs of Auxiliary Lemmas}
\subsection{Subexponentiality of square of sub-Gaussian random variables}
\begin{lemma}
	\label{lem:y^2 is subexponential} 
	If $y$ is a sub-Gaussian random variable with $0$ mean and parameter $1$, then $y^2$ is a sub-exponential random variable with parameters $(4\sqrt{2}, 4)$.
\end{lemma}	
\begin{proof}
	Since $y$ is a $0$ mean sub-Gaussian random variable with parameter $1$, we can write
	\begin{align*}
	(\forall \lambda \in \real) \E{\exp(\lambda y)} \leq \exp(\frac{\lambda^2}{2}) 
	\end{align*}
	Let $\Gamma(r)$ be the Gamma function, then moments of the sub-Gaussian variable $y$ are bounded as follows:
	\begin{align*}
	(\forall r \geq 0) \E{|y|^r} \leq r 2^{\frac{r}{2}}\Gamma(\frac{r}{2})
	\end{align*}
	Let $v \triangleq y^2$ and $\mu_v \triangleq \E{v}$. Using power series expansion and noting that $\Gamma(r) = (r-1)!$ for an integer $r$, we have:
	\begin{align*}
	\begin{split}
	\E{\exp(\lambda (v - \mu_v))} &= 1 + \lambda \E{v - \mu_v} + \sum_{r=2}^\infty \frac{\lambda^r \E{(v - \mu_v)^r}}{r!} \\
	&\leq 1 + \sum_{r=2}^\infty \frac{\lambda^r \E{|y|^{2r}}}{r!} \\
	&\leq 1 + \sum_{r=2}^\infty \frac{\lambda^r 2r 2^{r}\Gamma(r)}{r!} \\
	&= 1 + \sum_{r=2}^\infty \lambda^r 2^{r+1}\\
	&= 1 + \frac{8 \lambda^2}{1 - 2 \lambda}
	\end{split}
	\end{align*} 
	We take $\lambda \leq \frac{1}{4}$. Thus,
	\begin{align*}
	\begin{split}
	\E{\exp(\lambda (v - \mu_v))} &\leq 1 + 16\lambda^2 \\
	&\leq \exp(16 \lambda^2) \\
	&\leq \exp(\frac{(4\sqrt{2})^2\lambda^2}{2})
	\end{split}
	\end{align*}
	It follows that $v = y^2$ is a subexponential random variable with parameters $(4\sqrt{2}, 4)$. 
\end{proof}

\subsection{Subexponentiality of product of independent sub-Gaussian random variables}
\begin{lemma}
	\label{lem:subexponential}
	Let $p$ and $q$ be two independent sub-Gaussian random variables with $0$ mean and parameter $1$, then $pq$ is a sub-exponential random variable with parameters $(4\sqrt{2},4)$.
\end{lemma}
\begin{proof}
	Since $p$ and $q$ are both  $0$ mean sub-Gaussian random variable with parameter $1$, we can write
	\begin{align*}
	\begin{split}
	(\forall \lambda \in \real) \E{\exp(\lambda p)} &\leq \exp(\frac{\lambda^2}{2}) \\
	(\forall \lambda \in \real) \E{\exp(\lambda q)} &\leq \exp(\frac{\lambda^2}{2}) 
	\end{split}
	\end{align*}
	Let $\Gamma(r)$ be the Gamma function, then moments of the sub-Gaussian variable $p$ and $q$ are bounded as follows:
	\begin{align*}
	(\forall r \geq 0) \E{|p|^r} \leq r 2^{\frac{r}{2}}\Gamma(\frac{r}{2}) \\
	(\forall r \geq 0) \E{|q|^r} \leq r 2^{\frac{r}{2}}\Gamma(\frac{r}{2})
	\end{align*}
	Let $v \triangleq pq$. Note that $\E{v} = \E{pq} = \E{p}\E{q} = 0 $ due to independence. Using power series expansion and noting that $\Gamma(r) = (r-1)!$ for an integer $r$, we have:
	\begin{align*}
	\begin{split}
	\E{\exp(\lambda v )} &= 1 + \lambda \E{v} + \sum_{r=2}^\infty \frac{\lambda^r \E{v^r}}{r!} \\
	&\leq 1 + \sum_{r=2}^\infty \frac{\lambda^r \E{|p|^{r} |q|^{r}}}{r!} \\
	&\leq 1 + \sum_{r=2}^\infty \frac{\lambda^r \E{|p|^{r}} \E{|q|^{r}}}{r!} \\
	&\leq 1 + \sum_{r=2}^\infty \frac{\lambda^r r^2 2^r\Gamma(\frac{r}{2})^2}{r!}
	\end{split}
	\end{align*} 
	Note that $\Gamma(\frac{r}{2})^2 \leq \Gamma(r)$. Thus,
	\begin{align*}
	\begin{split}
	\E{\exp(\lambda v )} &\leq 1 + \sum_{r=2}^\infty \frac{\lambda^r r^2 2^r\Gamma(r)}{r!} \\
	&= 1 - \frac{8 (\lambda - 1) \lambda^2 }{(1 - 2 \lambda)^2} \\
	&\leq \exp(16 \lambda^2) \\
	&\leq \exp(\frac{(4\sqrt{2})^2}{2} \lambda^2)
	\end{split}
	\end{align*} 
	where last inequality holds for $|\lambda| \leq \frac{1}{4}$. Thus, $pq$ is subexponential with parameters $ (4\sqrt{2}, 2)$.
\end{proof}

\subsection{Norm Inequalities}
\label{sec:norm inequalities}
Here we will derive some norm inequalities which we will use in our proofs.
\begin{lemma}[Norm Inequalities]
	\label{lem:norminequalities}
	Let $\mb{A}$ be a row-partitioned block matrix which consists of $p$ blocks where block $\mb{A}_i \in \real^{m_i \times n}, \ \forall i \in \seq{p} $ and $\mb{B} \in \real^{n \times o}$. Then the following inequalities hold:
	\begin{align*}
	\|\mb{A} \mb{B} \|_{\BIF} \leq  \| \mb{A}\|_{\BIO} \| \mb{B}\|_{\infty,2}
	\end{align*} 
	\begin{align*}
	\|\mb{A} \mb{B} \|_{\BIO} \leq  \| \mb{A}\|_{\BIO} \| \mb{B}\|_{\infty,\infty}
	\end{align*} 
\end{lemma}
\begin{proof}
	Let $\f(.)$ be an operator which flattens the matrix and converts it to a vector. Let $\mb{Y}$ be a row-partitioned block matrix with same size and block structure as $\mb{A}$. 
	\begin{align*}
	\begin{split}
	\|\mb{A} \mb{B}\|_{\BIF} &= \max_{i \in \seq{p} } \|\f((\mb{A}\mb{B})_i)\|_2 \\
	&= \max_{i \in \seq{p} , \| \f(\mb{Y}_i) \|_2 \leq 1} \f((\mb{A}\mb{B})_i)^\T \f(\mb{Y}_i) \\
	&= \max_{i \in \seq{p} , \| \f(\mb{Y}_i) \|_2 \leq 1} [(\mb{A}_i)_{1.}\mb{B} \dots (\mb{A}_i)_{m_i.}\mb{B}] \f(\mb{Y}_i) \\
	&= \max_{i \in \seq{p} , \| \f(\mb{Y}_i) \|_2 \leq 1} [(\mb{A}_i)_{1.}\mb{B} (\mb{Y}_i)_{1.} + \dots + (\mb{A}_i)_{m_i.}\mb{B} (\mb{Y}_i)_{m_i.}] \\
	&\leq  \max_{\substack{i \in \seq{p} ,\\ \| \f(\mb{Y}_i) \|_2 \leq 1}} \| (\mb{A}_i)_{1.} \|_1 \| \mb{B} (\mb{Y}_i)_{1.} \|_{\infty} + \dots +\| (\mb{A}_i)_{m_i.} \|_1 \| \mb{B} (\mb{Y}_i)_{m_i.} \|_{\infty} \\
	&\leq \| \mb{A}\|_{\BIO} \| \mb{B}\|_{\infty,2}
	\end{split}
	\end{align*} 
	
	We follow a similar procedure for the last norm inequality.
	\begin{align*}
	\begin{split}
	\|\mb{A} \mb{B}\|_{\BIO} &= \max_{i \in \seq{p} } \|\f((\mb{A}\mb{B})_i)\|_1 \\
	&= \max_{i \in \seq{p} , \| \f(\mb{Y}_i) \|_\infty \leq 1} \f((\mb{A}\mb{B})_i)^\T \f(\mb{Y}_i) \\
	&= \max_{i \in \seq{p} , \| \f(\mb{Y}_i) \|_\infty \leq 1} [(\mb{A}_i)_{1.}\mb{B} \dots (\mb{A}_i)_{m_i.}\mb{B}] \f(\mb{Y}_i) \\
	&= \max_{i \in \seq{p} , \| \f(\mb{Y}_i) \|_\infty \leq 1} [(\mb{A}_i)_{1.}\mb{B} (\mb{Y}_i)_{1.} + \dots + (\mb{A}_i)_{m_i.}\mb{B} (\mb{Y}_i)_{m_i.}] \\
	&\leq  \max_{i \in \seq{p} , \| \f(\mb{Y}_i) \|_\infty \leq 1} \| (\mb{A}_i)_{1.} \|_1 \| \mb{B} (\mb{Y}_i)_{1.} \|_{\infty} + \dots \| (\mb{A}_i)_{m_i.} \|_1 \| \mb{B} (\mb{Y}_i)_{m_i.} \|_{\infty} \\
	&\leq \| \mb{A}\|_{\BIO} \| \mb{B}\|_{\infty,\infty}
	\end{split}
	\end{align*} 	
\end{proof}

\end{document}